\documentclass[11pt, lettersize]{article}

\usepackage[UKenglish]{babel}
\usepackage[a4paper,margin=1in]{geometry}

\usepackage{amsmath,amssymb,amsthm}
\usepackage{graphicx}
\usepackage{xcolor}
\usepackage{hyperref}
\usepackage[nameinlink,noabbrev]{cleveref}

\usepackage{booktabs}
\usepackage{algorithm}
\usepackage{algorithmic}
\usepackage{todonotes}
\usepackage{changepage}
\usepackage{cite}
\usepackage{comment}
\usepackage{enumitem}
\usepackage{moresize}
\usepackage{pifont}
\usepackage{xspace}
\usepackage{lineno}
\usepackage{thm-restate}

\title{The Contiguous Art Gallery Problem is in $\Theta(n\log n)$}

\author{
Sarita de Berg\thanks{IT University of Copenhagen, Denmark. \texttt{debe@itu.dk}}
\and
Jacobus Conradi\thanks{University of Copenhagen, Denmark. \texttt{jaco@di.ku.dk}}
\and
Ivor van der Hoog\thanks{IT University of Copenhagen, Denmark. \texttt{ivva@itu.dk}}
\and
Eva Rotenberg\thanks{IT University of Copenhagen, Denmark. \texttt{erot@itu.dk}}
}

\date{}

\newtheorem{problemstatement}{Problem Statement}
\newtheorem{datastructure}{Data Structure}
\newtheorem*{invariants}{Invariants}
\newtheorem{invariant}{Invariant}
\newtheorem*{properties}{Properties}

\graphicspath{{figures/}}

\newcommand{\badSymbol}{\star\xspace}   
\newcommand{\uglySymbol}{\ding{247}\xspace}

\newcommand{\realRAM}{\textnormal{\texttt{realRAM}}\xspace}
\newcommand{\wordRAM}{\textnormal{\texttt{wordRAM}}\xspace}

\newcommand{\ts}{\textsuperscript}

\newcommand{\core}[2]{\ensuremath{\mathcal{E}[#1,#2]}}
\newcommand{\reducedDom}{\ensuremath{\mathcal{D}}\xspace}
\newcommand{\nex}[1]{\ensuremath{\textnormal{\texttt{next}}(#1)}}
\newcommand{\nexFunc}{\ensuremath{\textnormal{\texttt{next}}}\xspace}

\newcommand{\short}[2]{\ensuremath{S(#1,#2)}}
\def\polylog{\operatorname{polylog}}
\newcommand{\slidingds}{\textsc{Sliding-Window} data structure~(\ref{ds:sliding})\xspace}
\newcommand{\setInter}{\textsc{SetDisjointness}\xspace}
\newcommand{\contArt}{\textsc{Contiguous Art Gallery}\xspace}
\newcommand{\artGall}{\textsc{Art Gallery}\xspace}
\newcommand{\baddomsigma}{\ensuremath{\mathcal{B}_\sigma}\xspace}
\newcommand{\baddom}{\ensuremath{\mathcal{B}}\xspace}

\newcommand{\mysubpara}[1]{%
  \par\vspace{0.3\baselineskip}%
  \noindent\textsf{\textbf{#1}}\hspace{0.5em}%
}

\newtheorem{lemma}{Lemma}
\newtheorem{observation}{Observation}
\newtheorem{definition}{Definition}
\newtheorem{theorem}{Theorem}
\newtheorem{corollary}{Corollary}
\newtheorem{remark}{Remark}

\begin{document}
\nolinenumbers

\maketitle

\begin{abstract}
Recently, a natural variant of the \textsc{Art Gallery} problem, known as the \textsc{Contiguous Art Gallery} problem was proposed. 
Given a simple polygon $P$, the goal is to partition its boundary $\partial P$ into the smallest number of contiguous segments such that each segment is completely visible from some point in~$P$. 
Unlike the classical \textsc{Art Gallery} problem, which is NP-hard, this variant is polynomial-time solvable. 
Three independent submissions to SoCG 2025 presented algorithms for this problem, each achieving a running time of $O(k n^5 \log n)$ (or $O(n^6\log n)$), where $k$ is the size of an optimal solution. 
Interestingly, these results were obtained using entirely different approaches, yet all led to roughly the same asymptotic complexity.

In the \realRAM-model, the prevalent model in computational geometry, we present an $O(n \log n)$-time algorithm, achieving an $O(k n^4)$ factor speed-up over the previous state-of-the-art. 
We also give a straightforward lower bound reduction from the set intersection problem. 
We thus show that the Contiguous Art Gallery problem is in $\Theta(n \log n)$.
\end{abstract}

\paragraph{Funding.} This work was supported by the VILLUM Foundation grant (VIL37507) ``Efficient Recomputations for Changeful Problems'' and by the Carlsberg Foundation, grant CF24-1929.

\paragraph{Acknowledgements.} The authors wish to thank Frank Staals and Jack Spalding-Jamieson for their helpful discussions.

\thispagestyle{empty}
\setcounter{page}{0}

\newpage
\section{Introduction}\label{sec:intro}

\setcounter{page}{1}

The \artGall problem is a classical problem in computational geometry. 
Given a simple polygon $P$ with $n$ vertices, the task is to compute the smallest set of guards such that every point $p \in P$ is visible to at least one guard. 
The problem was first posed by Klee in 1973 and later formalized by Chvátal~\cite{chvatal1975combinatorial}, who proved that $\lfloor n/3 \rfloor$ guards always suffice and that this combinatorial bound is tight.
Since its introduction, the problem has been extensively studied. 
O’Rourke and Supowit~\cite{Rourke1983hard} showed the problem is NP-hard if $P$ is allowed to have holes. Lee and Lin~\cite{lee1986VertexNPhard} strengthened this by proving NP-hardness even when~$P$ is a simple polygon. 
Eidenbenz, Stamm, and Widmayer~\cite{EidenbenzStammWidmayer2001} established APX-hardness, and Bonnet and Miltzow~\cite{Bonnet2020Parametrized} showed the problem is W[1]-hard parametrized by the number of guards. 
Finally, Abrahamsen, Adamaszek, and Miltzow~\cite{abrahamsen2018ETRhard} proved that the problem is $\exists \mathbb{R}$-complete.

\mysubpara{Problem variants.}
The \artGall problem has many variants. We briefly review three types of variants to make two points: that \emph{many} such derivatives are actively studied, and that they are often computationally hard or only have known high-degree polynomial-time solutions.

The first family of variants restricts the structure of the polygon $P$. 
Krohn and Nilsson~\cite{KrohnNilsson2012} showed that the problem remains NP-hard when $P$ is $x$-monotone. 
Schuchardt and Hecker~\cite{SchuchardtHecker1995} established NP-hardness for orthogonal polygons, and Tom\'{a}s~\cite{Tomas2013ThinOrth} extended this to the case of thin orthogonal polygons. 
King and Krohn~\cite{15D_terrain} proved NP-hardness even when $P$ is restricted to be a 1.5D terrain.
A second line of work constrains guard placement.  
Lee and Lin~\cite{lee1986VertexNPhard} proved NP-hardness when guards are restricted to the polygon vertices. 
Rieck and Scheffer~\cite{Rieck2024Dispersive} studied a dispersive variant where guards must maintain a minimum distance from one another, and showed NP-hardness as well. 
%TODO: Put this back in!
%Related is the \emph{watchman problem}, in which one or more guards move along paths inside the polygon. 
%This problem also admits several sub-variants, the majority of which are NP-hard~\cite{ChinNtafos1988, DumitrescuToth2012, Mitchell2013}. 
A third family of variants modifies the notion of visibility. 
Lee and Lin~\cite{lee1986VertexNPhard} considered \emph{edge guards}, where a point of $P$ is visible if it can be seen from some point on a guarding edge. 
Mahdavi, Seddighin, and Ghodsi~\cite{MAHDAVI2020163} studied edge guards with orthogonal vision that can see through up to $k$ walls for rectilinear polygons in which new guarding edges may be introduced in $P$. 
Biedl and Mehrabi~\cite{Biedl2017Mehrabi} analyzed rectilinear visibility, where two points see each other if there exists a rectilinear path between them contained in the orthogonal polygon $P$ that may contain holes. All of these variants were shown to be NP-hard.
Motwani, Raghunathan, and Saran~\cite{s-covering} showed that this problem with rectilinear vision can be solved in $O(n^{8})$ time in a rectilinear polygon without holes.
Worman and Keil~\cite{Worman2007Decomposition} considered rectangle vision, where two points $p,q \in P$ see each other if there exists a rectangle contained in $P$ having $p$ and $q$ as opposite corners; they obtained an $\tilde{O}(n^{17})$ polynomial-time algorithm.

\mysubpara{Low-polynomial variants.}
Only very few \artGall variants are solvable in low-polynomial time, and all known results rely on very strong structural restrictions. 
Lee and Preparata~\cite{Lee1979ptimal} gave a linear-time algorithm to decide whether a polygon can be guarded by a single guard. 
De Berg, Durocher, and Mehrabi~\cite{DBLP:journals/jda/BergDM17} presented a linear-time algorithm to minimize the number of edge guards under orthogonal vision in a monotone rectilinear polygon. 
Palios and Tzimas~\cite{Palios2013-star-cover} restrict each guard’s visibility region to an orthogonal $r$-star and restrict the polygon to class-3 orthogonal polygons, and obtained near-linear time. 
Biedl and Mehrabi~\cite{BiedlMehrabi2016_rGuardingThinOrthogonal} combined several restrictions---$P$ must be orthogonal~\cite{SchuchardtHecker1995}, thin~\cite{Tomas2013ThinOrth}, and rectilinear vision~\cite{Biedl2017Mehrabi}.

\mysubpara{\contArt.}
Laurentini~\cite{Laurentini1999GuardingWalls} suggested guarding only the boundary of the polygon, based on the observation that in most galleries, artworks are displayed on walls. This leads to the natural question:
`Can~we compute a smallest set of guards such that every point on $\partial P$ is visible to at least one guard?' 
This problem is NP-hard~\cite{Laurentini1999GuardingWalls}, and Stade~\cite{stade2025interiorBoundaryHard} recently proved that it is $\exists \mathbb{R}$-complete. 
Thomas Shermer proposed the closely-related \contArt problem. 
Here, each guard $g$ can only guard a contiguous segment $[u, v]$ of $\partial P$. 
The goal is to compute the smallest set of contiguous guards which together see $\partial P$ (see Figure~\ref{fig:problem_def}).
At SoCG~2025, three independent \emph{submissions}, which were merged into one paper, presented \realRAM solutions with comparable running times but entirely different methods. Since their techniques were not merged, there now exist three ways to solve this problem:

\begin{figure}[t]

    \centering
   \includegraphics{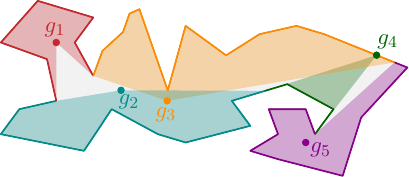}
   \caption{A simple polygon $P$ and five contiguous guards that guard the entire boundary of $P$.}
    \label{fig:problem_def}
\end{figure}

\begin{itemize}[noitemsep]
\item Merrild, Rysgaard, Schou, and Svenning~\cite{merrild2024contiguous} gave an $O(k n^5 \log n)$-time greedy algorithm. 
For any point $u \in \partial P$, let $\nex{u}$ denote the farthest point $v$ along $\partial P$ such that there exists a guard that sees $[u,v]$. 
Starting from $u$, recursively apply $\nexFunc$ until the resulting set $\{ [u, v] \}$ covers $\partial P$. 
Such a sequence defines a \emph{revolution}, which they show yields a solution of size at most $k+1$.  
They compute a revolution in $O(n^2 \log n)$ time and obtain after $O(k n^3)$ revolutions an optimal solution, leading to an $O(k n^5 \log n)$-time algorithm.
    \item  The authors of~\cite{biniaz2024contiguous} constructed a candidate set $\mathbb{C}$ of $O(n^4)$ guards and guaranteed that there exists an optimal solution using a guard from $\mathbb{C}$. 
They compute $\mathbb{C}$, and brute force for each $g \in \mathbb{C}$ the revolution, using $O(k n^5 \log n)$ total time.
\item 
Robson, Spalding-Jamieson, and Zheng~\cite{RobsonSpaldingZheng2024_AnalyticArcCover} took an analytical approach based on the same $\nex{u}$ function. 
They parametrize the boundary $\partial P$ by $[1,2n+1)$, where points $i$ and $i+n$ correspond to the same vertex of $P$. Then we can treat $\nexFunc$ as a function from $[1,n+1)$ to $[1,2n+1)$. 
 $\partial P$ can be partitioned into $O(n^3)$ intervals, where $\nexFunc$ is of constant-complexity in each interval. They extend this to an $O(n^6 \log n)$-time algorithm.
\end{itemize}

\mysubpara{Contribution.}
At first sight, \contArt appears to follow the same pattern as most \artGall variants: guarding $\partial P$ in general is hard, and the restricted version admits high-polynomial solutions. 
Yet, we show that we can significantly improve the running time to $O(n \log n)$.
Our \realRAM algorithm uses linear space and achieves an $O(k n^4)$ factor speed-up over the previous state-of-the-art\footnote{A preliminary version of this paper, which we cannot cite without breaking blindness, achieves  $\tilde{O}(k n^2)$ time.}
We further prove that this bound is optimal via a reduction from sorting, which in the \realRAM has an $\Omega(n \log n)$ lower bound. 
Our result thus provides a significant speedup and tight bounds for the \contArt problem.

\mysubpara{Technical contribution.}
First, we explore the extent to which previous techniques can solve this problem.  
Regarding~\cite{biniaz2024contiguous}, we show the existence of a linear-size candidate set $\mathbb{C}$ which we can compute in near-linear time, yielding a $\tilde{O}(kn^2)$-time algorithm (see the footnote). 
What remains is to remove the brute force approach to compute revolutions. 
We show a sweep-line algorithm that, for a set of points $X \subset \partial P$, computes $\nex{x}$ for all $x \in X$ in $O(n \log n)$ total time.
Applying this second sweep-line algorithm $k$ times yields an $O(k n \log n)$-time algorithm. 

To remove the remaining factor of $k$, we first select an arbitrary $x \in X$ and compute its revolution. If its revolution has $k'$ guards, we know that the optimal solution is either $k'$ or $k' - 1$. 
What remains is to \emph{implicitly} apply the $\nexFunc$ function to all $x \in X$, $k' - 1$ times and verify whether there is a point that then overtakes itself. 
This is the most technical part of this paper, which incorporates the algebraic view of the $\nexFunc$ function introduced by Robson, Spalding-Jamieson, and Zheng~\cite{RobsonSpaldingZheng2024_AnalyticArcCover}. They  showed that $\nexFunc$ can be represented as a piecewise function with $O(n^3)$ constant-complexity pieces. 
We show, through a much more technical argument, that the $\nexFunc$ function can be represented as a piecewise function with only $O(n)$ constant-complexity pieces, and that this representation can be computed in $O(n \log n)$ total time. 
Given $k'$, we  create a simultaneous-evaluation technique that maintains $X$ in a balanced tree $T$ and implicitly applies the $\nexFunc$ function to contiguous subtrees of $T$.
We use this implicit representation to verify if some $x \in X$ has a revolution of size $k' - 1$, yielding the optimal solution.

\section{Preliminaries}\label{sec:prelim}
Our input is a simple polygon $P$: the closed region bounded by a simple closed curve consisting of $n$ vertices defining $n$ edges. 
A vertex of $P$ is a \emph{reflex vertex} if its interior angle exceeds $\pi$. We denote for two points $s,t \in P$ the shortest path from $s$ to $t$, which is a polygonal chain, by $\short{s}{t}$. 
We denote by $\partial P$ the boundary of $P$ and assume that $\partial P$ is given in counter-clockwise order. 
Each edge is an ordered pair $\overline{u\,v}$ following this counter-clockwise ordering, i.e., it is directed from $u$ to $v$.
A point lies \emph{strictly to the right} of $\overline{u\,v}$ if it lies in the open half-plane to the right of the directed line through $u$ and $v$. 
A point \emph{to the right} of $\overline{u\,v}$ may lie on this line. 

\mysubpara{Parameterizing $\partial P$.} 
We assume that for any integer $i \in [1, n-1]$, the consecutive vertices $v_i$ and $v_{i+1}$ of $P$ appear counter-clockwise. 
Consequently, the interior of $P$ lies immediately left of $\overline{ v_i \, v_{i+1}}$. 
For convenience, we define a continuous surjective function $[1, 2n + 1) \rightarrow \partial P$ such that for every vertex $v_j$, both $j$ and $n + j$ map to $v_j$. 
Thus, each point on $\partial P$ can be represented by two real values: one in $[1, n + 1)$ and one in $[n + 1, 2n + 1)$. 
Using this parametrization, we define (open) chains as follows.
A \emph{chain} is a sequence of edges. 
For two points $u, v \in \partial P$, we denote by $[u, v]$ the chain obtained by traversing $\partial P$ from $u$ to $v$ in counter-clockwise order. 
We denote by $(u, v)$ the \emph{open chain}, consisting of all points $x \in [u, v]$ with $x \neq u$ and $x \neq v$.
The chain between any two points on $\partial P$ can be described as $[u, v]$ with $u \in [1, n + 1)$ and $v \in [u, 2n + 1)$.

\mysubpara{Problem statement.}
A point $x$ \emph{sees} a point $y$ if the segment $\overline{x\,y}$ is contained in $P$. 
A \emph{(contiguous) guard} is a tuple $(g, [u, v])$ consisting of a point $g \in P$ and a chain $[u, v]\subseteq\partial P$ such that every point on the chain is visible from $g$. 
The formal problem statement is:  Given a simple polygon $P$ with $n$ vertices, compute a minimum-size set of contiguous guards $G$ such that their corresponding chains cover the entire boundary interval $[1, n + 1]$.

\mysubpara{Domination}
We say that a guard $(g, [u, v])$ \emph{dominates} another guard $(g', [u', v'])$ if $[u', v'] \subseteq [u, v]$. 
It \emph{strictly dominates} $(g', [u', v'])$ if $[u', v'] \subsetneq [u, v]$. 
For a fixed $u \in [1, n + 1)$, we frequently compute a \emph{maximal} $v \in [u, 2n + 1)$ such that there exists a point $g \in P$ for which $(g, [u, v])$ is a guard.
By this we mean that for all other guards $(g', [u, v'])$, it holds that $v'  \in [u, v]$.

\mysubpara{Visibility core.}
For any chain $[u, v] \subseteq \partial P$, we define the \emph{visibility core} $\core{u}{v}$ as the set of points $p \in \mathbb{R}^2$ that lie left of all edges in $[u, v]$. By the following, we may assume that $k \geq 2$:

\begin{observation}[By~\cite{Lee1979ptimal}]\label{obs:leftOfEdges}\label{obs:trivialOneCenter}
If $(g, [u, v])$ is a guard, then $g$ lies left of all edges of $P$ intersecting the open chain $(u, v)$. 
In particular, $g$ lies in $\core{u}{v}$.
Moreover, there exists a guard that can see all of $\partial P$ if and only if $\core{1}{n+1}$ is non-empty which can be computed in linear time.
\end{observation}

\section{Technical Overview}
Our primary contribution is a tight analysis of the \contArt problem in the \realRAM-model. 
We present an $O(n \log n)$-time algorithm, achieving a speedup by a factor of $\tilde{O}(k n^4)$ compared to previous works~\cite{theEnemy,merrild2024contiguous,biniaz2024contiguous,RobsonSpaldingZheng2024_AnalyticArcCover}.
We complement this by a relatively simple lower bound of $\Omega(n \log n)$ in the \realRAM-model. 
Our improvements arise from three key ideas:

For any $x \in [1, n + 1)$, let $\nex{x}$ denote the maximal value $v$ such that there exists a guard $(g, [u, v])$. 
We define a \emph{revolution} of $x$ as the sequence of guards obtained by applying $\nexFunc$ recursively until the point overtakes $x$. 
Our first contribution is to construct a set $X$ of points along $\partial P$ with the property that there exists some $x \in X$ for which its revolution yields an optimal solution.
We construct a set $X$ of linear size and compute it in $O(n \log n)$ time which already yields a $\tilde{O}(k n^2)$-time algorithm using the revolution techniques from~\cite{theEnemy,merrild2024contiguous,biniaz2024contiguous}.

Our second idea is a sliding-window algorithm that, for an ordered pair of indices $(i, j)$, repeatedly increments $i$ or~$j$ until $i = n + 1$. 
We use the sliding window to compute $\nex{x}$ for all $x \in X$ in $O(n \log n)$ total time using only $O(n)$ space. 
By applying the sliding-window procedure $k$ times, recursively, starting with $X$ we obtain an $O(k n \log n)$-time algorithm.

Our final idea is by far the most technical, and deviates most from~\cite{theEnemy,merrild2024contiguous,biniaz2024contiguous,RobsonSpaldingZheng2024_AnalyticArcCover}. 
We partition $\partial P$  into a \emph{linear} number of intervals such that, for each interval $[x_1, x_2]$, the function $\nex{x}$ on the domain $[x_1, x_2]$ can be described by a constant-complexity function.
Given this function representation, we show how to lazily apply and implicitly evaluate the \nexFunc function to all $x \in X$.
For some arbitrary $x \in X$ we compute its revolution and suppose that it has size $k'$. 
Then either $k = k'- 1$ or $k = k'$. 
By our lazy evaluation technique, we show how to verify in $O(n \log n)$ time whether there exists a $\hat{x} \in X$ whose revolution has size $k' - 1$. 
Given $\hat{x}$, we can compute the optimal solution by computing its revolution, using $O(n \log n)$ total time. 

\subsection{Upper bound}

For any $x \in [1, n + 1)$, let $\nex{x}$ denote the maximal value $v \in [x, x+n]$ such that there exists a guard $(g, [u, v])$ and denote by $\nexFunc^i(x)$ the result of recursively applying this function $i$ times. 
Our goal is to compute a set $X \subseteq \partial P$ such that for at least one $x \in X$, $\nexFunc^k(x) \geq x + n$.

\mysubpara{Classifying guards.}
\Cref{sec:classification} introduces the structural results to compute $X$ and the $\nexFunc$ function. 
We define three collections of highly structured guards such that every guard is dominated by one of these three types. 
Specifically, we partition the set of all guards into \emph{good} and \emph{bad} guards. 
Intuitively, a guard $(g, [u, v])$ is \emph{good} if the angle $\sphericalangle(u, g, v)$ is at most~$\pi$, and \emph{bad} otherwise. 
We then define three \emph{dominators}: the good, the bad, and the ugly. 

Good dominators are well-structured: for any ordered pair of indices $(i, j)$ there exists at most one good dominator, implying that there are $O(n^2)$ such dominators. Moreover, every good guard is dominated by a good dominator. 
Bad dominators are slightly less structured: for each ordered pair $(i, j)$, we consider the polygon $\core{i-1}{j+1}$ and define a guard $(g, [u_{\max}, v_{\max}])$ for each vertex $g$ of $\core{i-1}{j+1}$, where $[u_{\max}, v_{\max}]$ is some maximal chain visible to $g$. 
Since $\core{i-1}{j+1}$ has $O(n)$ vertices, there are $O(n^3)$ bad dominators. The set of ugly dominators is not as susceptible to discretization. Instead, we prove that each bad guard is dominated by either a bad or an ugly dominator, and that ugly dominators can be computed on the fly. 

We can now provide the intuition for why $\nex{x}$ can be evaluated in logarithmic time. Precompute all good and bad dominators $(g, [u, v])$ and store their corresponding intervals $[u, v]$ in a segment tree $T$. 
For any $x \in \partial P$, perform a stabbing query on $T$ and return the maximal right endpoint $r$ among the intervals stabbed by $x$. 
Let $v = \nex{x}$. 
The guard $(g, [u, v])$ is by maximality either a good, bad or ugly dominator. 
If it is a good or bad dominator, then $v = r$.
Compute the ugly dominator for $x$ on the fly and compare it to $r$ to obtain $\nex{x}$. 

For this approach to run in $O(n \log n)$ time, the segment tree must contain only $O(n)$ dominators. 
To ensure this, we define the \emph{reduced good dominators} $\reducedDom$ as those good dominators that are not \emph{strictly} dominated by another good dominator. 
Similarly, we define the \emph{reduced bad dominators} $\baddom$ as the bad dominators that are not \emph{strictly} dominated by \emph{any} other guard. 
These distinct definitions matter because the reduced good dominators serve a dual purpose:

\mysubpara{Defining a set \boldmath$X$.}
The reduced good dominators $\reducedDom$ are not only essential for computing $\nexFunc$, but they also induce the set $X$. 
Specifically, we prove that there always exists an optimal solution $\textsc{Opt}$ to the \contArt problem such that at least one guard $(g, [u, v]) \in \textsc{Opt}$ either belongs to $\reducedDom$ or has $u$ as a vertex of $P$. 
We define $X$ as the set of all vertices of $P$, together with all points $u \in \partial P$ such that there is a reduced good dominator $(g, [u, v])$. 

\mysubpara{Sliding windows.} 
%First, we compute a linear-size superset $\reducedDom'$ that contains $\reducedDom$ in $O(n \log n)$ time via a sliding window algorithm over the indices $(i, j)$ that define a good dominator. 
We use a sliding window over indices $(i,j)$ to compute (a superset of) $\reducedDom$.
The remainder of our %O(n \log n)$-time 
algorithms rely on a different but very specific sliding window: 

\begin{restatable}{definition}{slidingsequence}
    \label{def:sequence}
    We define a \emph{sliding sequence} as an ordered set of index pairs $\{(i, j)\}$ of linear size where, for every pair of consecutive elements $((i, j), (i', j'))$, we have $(i' - i,\, j' - j) \in \{(1, 0), (0, 1)\}$.  
    Given a polygon $P$, a sliding sequence $\sigma$ is said to be \emph{conforming} if for each $u \in [1, n + 1)$, with $u \in [i - 1, i)$ and $\nex{u} \in (j, j + 1]$, then $(i, j) \in \sigma$.
\end{restatable}

Let $\sigma$ be a conforming sliding sequence.
Each bad dominator has a defining index pair $(i, j)$. Since reduced bad dominators are not strictly dominated by \emph{any} other guard, every guard in $\baddom$ has its corresponding index pair $(i, j) \in \sigma$. 
From this observation, we upper-bound $|\baddom|$.  
Consider iterating over all $(i, j) \in \sigma$. 
By point-line duality, maintaining the convex visibility core $\core{i-1}{j+1}$ during this iteration corresponds to maintaining the convex hull of a point set under first-in–first-out updates. 
Chan, Hershberger, and Pratt~\cite{Chan2019Time-Windowed} show that the convex hull of such an update sequence of length $n$ has $O(n)$ vertices.
This implies an algorithm to dynamically maintain $\core{i-1}{j+1}$ subject to incrementing $i$ and $j$, and that $|\baddom| \in O(n)$.
It remains to compute such a $\sigma$, and derive the sets $\reducedDom$ and $\baddom$, in $O(n \log n)$ time.

\mysubpara{Computing $\sigma$ and the reduced dominators.} In \Cref{sec:computingDiscreteDominators} we compute linear-size sets $\reducedDom'$ and $\baddomsigma$ where $\reducedDom \subseteq \reducedDom'$ and $\baddom \subseteq \baddomsigma$.
We first compute $\reducedDom'$ using a sliding-window algorithm. 
Note that we cannot use a conforming sliding sequence $\sigma$ for this task, as a reduced good dominator may be strictly dominated by a guard that is not a good dominator. 
Consequently, the defining pair $(i, j)$ for a guard $(g, [u, v]) \in \reducedDom$ need not belong to $\sigma$. 

For all subsequent algorithms, we use a conforming sliding sequence $\sigma$, which we compute by starting from $(i, j) = (1, 1)$ and repeatedly incrementing either $i$ or $j$ using  a simple rule:  
while there exists a guard $(g, [u, v])$ such that $[i, j + 1] \subseteq [u, v]$, we increment $j$ and add the new pair $(i, j)$ to $\sigma$;  
otherwise, we increment $i$ and add the new pair $(i, j)$ to $\sigma$.  
By~\cite{Chan2019Time-Windowed}, we can maintain $\core{i - 1}{j + 2}$ as a balanced binary tree while incrementing $i$ and $j$.  
We show that, given $i$, $j$, and $\core{i - 1}{j + 2}$, this decision can be made in $O(\log n)$ time, yielding an overall $O(n \log n)$ algorithm for constructing $\sigma$.
From $\sigma$, we derive a linear-size set $\baddomsigma$ containing $\baddom$. We store the visibility chains of guards in $\reducedDom'$ and $\baddomsigma$ in a segment tree $T$. Thus, we may obtain for any point $x \in \partial P$, the guard $(g, [u, v]) \in \reducedDom' \cup \baddomsigma$ that has $x \in [u, v]$ and maximizes $v$.

\mysubpara{An \boldmath$O(k n \log n)$ algorithm.} 
In \Cref{sec:knlogn}, we consider the set $X$ and iterate over all $(i, j) \in \sigma$. 
For each $x \in X$, we perform a stabbing query on $T$ to find the maximum value $r$ such that there exists a guard $(g, [x, r]) \in \reducedDom' \cup \baddomsigma$. 
The value $\nex{x}$ differs from $r$ only if it is realised by an ugly dominator. 
The value $\lfloor r \rfloor$ serves as a lower bound for the index $j$ such that $\nex{x} \in (j, j + 1]$, and by definition $(i, j) \in \sigma$. 
These observations yield an algorithm that iterates over $\sigma$ and computes $\nex{x}$ for all $x \in X$ in $O(n \log n)$ total time. 
Repeating this process $k$ times produces an optimal solution in $O(k n \log n)$ time. 
This is already a major improvement over the state-of-the-art~\cite{theEnemy,merrild2024contiguous,biniaz2024contiguous,RobsonSpaldingZheng2024_AnalyticArcCover} and serves as a crucial subroutine for our $O(n \log n)$-time algorithm.

\mysubpara{An \boldmath$O(n \log n)$ algorithm.} 
We represent $\nexFunc$ as a piecewise, possibly degenerate, M\"obius transformation. 
This approach was first proposed in~\cite{theEnemy,RobsonSpaldingZheng2024_AnalyticArcCover}, where the resulting function consisted of $O(n^3)$ pieces and was computed in $\tilde{O}(n^5)$ time. 
In contrast, we use our dominators to define a function with only $O(n)$ pieces, which can be computed in $O(n \log n)$ time.
Specifically, in \Cref{sec:functionArrangement}, we partition the boundary $\partial P$ into intervals $I$ such that either:
\begin{itemize}[nolistsep]
    \item a reduced good dominator $(g,[u,v])$ sees $[x,\nex{x}]$ for all $x\in I$, so $\nex{x}=v$;
    \item a reduced bad dominator $(g',[u',v'])$ sees $[x,\nex{x}]$ for all $x\in I$, so $\nex{x}=v'$; or
    \item for all $x\in I$, $\nex{x}=\frac{A+Bx}{C+Dx}$ for constants $A,B,C,D$ depending only on $I$.
\end{itemize}
As there are only $O(n)$ reduced good and bad dominators, the first two cases introduce at most $O(n)$ interval boundaries. 
We next prove that there are at most $O(n)$ intervals $\mathcal{U}$ where
for all $x \in \mathcal{U}$, no reduced good or bad dominator sees $[x, \nex{x}]$, i.e., only an ugly dominator does.

We first prove that the vertices of the shortest path from $x$ to $\nex{x}$ behave in a structured way, unless the optimal solution consists of at most $3$ guards. 
These vertices play a key role in defining the ugly dominator $g_u$, and the resulting structural properties guarantee both the existence and computability of the $O(n)$ M\"obius transformations in $O(n \log n)$ total time.
Finally, in \Cref{sec:nlogn}, we combine all components to obtain our $O(n \log n)$ algorithm. 
We apply our $O(k n \log n)$ algorithm to check whether a solution with $k \leq 3$ guards exists, in $O(n\log n)$ time. 
If not, we use a lazy evaluation technique to compute the smallest $k$ such that there exists a point $x \in X$ for which $\nexFunc^k(x)\geq x+n$, or conversely, $\partial P \subseteq [x,\nexFunc^k(x)]$.

\begin{figure}[b]

    \centering
    \includegraphics[scale=.8]{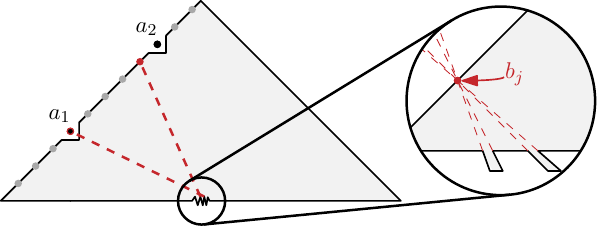}
\caption{The lower-bound construction. Black points belong to $A$, and red points to $B$. }
    \label{fig:sketch}
\end{figure}

\subsection{Lower bound}
In \Cref{sec:lowerBound}, we show a lower bound for the \contArt problem via a reduction from \textsc{Set Disjointness}, where we are given two sets $A, B \subseteq [0, n^3]$, each of size $n$, and must decide whether $A \cap B = \emptyset$. 
While in the \wordRAM there is a deterministic $O(n\log\log n)$-time algorithm for solving this problem (via sorting~\cite{Han2002Deterministic}), in comparison-based models of computation such as the \realRAM, it has a 
worst-case lower bound of $\Omega(n \log n)$, % in any comparison-based model of computation, 
even if $A$ is already sorted~\cite{YaoLowerBound}.
Given $A, B \subseteq [1, n^3]$ \hbox{with $A$ sorted, we can in linear} time construct a simple polygon $P$ such that $A \cap B = \emptyset$ if and only if $P$ can be guarded with $2|B| + 2$ guards (see \Cref{fig:sketch}). In particular, we specify $P$ by providing its edges in clockwise order. 

We begin with a %n axis-aligned 
triangle $T$ on the vertices $(0,0)$, $(n^3,n^3)$, $(2n^3,0)$, where the edge $t:$ $(0,0)$ to $(n^3,n^3)$ represents the interval $[0,n^3]$.
We construct the polygon in clockwise order, starting from $(0,0)$. 
For each of the increasing values $a_i \in A$, we
make an indent in $t$ 
ensuring that the point corresponding to $a_i$ does not lie in the polygon $P$. 
Next, we partition a small segment $(n^3,0)$ to $(n^3 + n, 0)$ on the bottom edge 
($(2n^3,0)$ to $(0,0)$) 
into $|B|$ intervals $I_1,I_2,\ldots$ with $I_j$ from $(n^3 + 0.5 j, 0)$ to $(n^3 + 0.5 j + 0.25j, 0)$.
We iterate over $j$ from low to high and introduce for $b_j \in B$ a \emph{pocket}: two consecutive extensions of $I_j$ such that a guard can see $I_j$ if and only if it stands at the point $(b_j,b_j) \in t$ 
Hence, there exists a guard that can cover $I_j$ if and only if $b_j \notin A$.
Note that this can be done even if $B$ is not sorted, as the placement of this gadget on $I_j$ depends on $j$ and not on $b_j$. 
Finally, we place a \emph{blocker} between consecutive intervals $I_j$ and $I_{j+1}$ such that each $I_j$ and each blocker requires a unique guard (ensuring that we are robust against repetitions in $B$).
The resulting polygon $P$ (specified by its edges in clockwise order) can be constructed in linear time and guarded with $2|B| + 2$ contiguous guards if and only if $A \cap B = \emptyset$, establishing an $\Omega(n \log n)$ lower bound.

\section{A combinatorial classification of guards}\label{sec:classification}
We say each guard is either \emph{good} or \emph{bad}. In Fig.~\ref{fig:problem_def}, $\{ g_2, g_3, g_4 \}$ are good while $\{ g_1, g_5 \}$ are bad.

\begin{definition}
    Let $u$ and $v$ be points on $\partial P$ with $u < v$. 
    A contiguous guard $(g, [u, v])$ is \emph{good} if $g \neq u$ and $g \neq v$ and the angle $\sphericalangle(u,g,v) \leq \pi$, and \emph{bad} otherwise.
\end{definition}

We will define three collections of guards called \emph{dominators}: the good, the bad, and the ugly. 
We prove that all good guards are dominated by a good dominator, which induces a set $X$ of $O(n)$ points along $\partial P$ such that there exists an optimal solution that includes a guard $(g, [x, v])$ for $x \in X$.  
We prove that all bad guards are dominated by either a bad or an ugly dominator (see Figure~\ref{fig:flowchart}). To this end, we first make some observations:

\begin{lemma}\label{lem:verticesAreGood}
    Let $(g, [u,v])$ be a guard where, for the fixed vertex $u$, $v$ is maximal. Then either $v$ is a reflex vertex of $P$ or $\overline{v \,g}$ contains a reflex vertex of $P$ in its interior. 
\end{lemma}
\begin{proof}
    %Note that $v$ is maximal w.r.t.~$g$ and $u$, and hence $g\neq v$, as otherwise $[u,v]$ could be extended by every edge containing $v$, contradicting the maximality of $v$. 
        By maximality of $v$, $g \neq v$ and there exists an $\varepsilon^* > 0$ such that $\forall \varepsilon \in (0, \varepsilon^*]$ the shortest path from $g$ to $v + \varepsilon$ visits a reflex vertex $x$. 
        If $x = v$ then $v$ is a reflex vertex of $P$. 
        Otherwise, $\overline{v \, g}$ contains $x$ in its interior. 
\end{proof}

\begin{lemma}
\label{cor:left-turning}
Let $(g, [u, v])$ be a bad guard, and suppose that $[u, v]$ is inclusion-wise maximal for this choice of $g$. Then the shortest path from $u$ to $v$ in $P$ is a left-turning chain, with at least one interior vertex if one of $u$ and $v$ is not a vertex of $P$, and at least two interior vertices if neither is a vertex of $P$.

The same left-turning property holds if $g$ and $u$ are fixed and $v$ is chosen maximal; in this case, the shortest path has at least one interior vertex whenever $v$ is not a vertex of $P$.
\end{lemma}

\begin{proof}
    For a guard that is inclusion-wise maximal it cannot be that $g = u$ or $g = v$. That the shortest path from $u$ to $v$ is left-turning follows from the fact that $g$ can see $u$ and $v$ and $\sphericalangle(u, g, v) > \pi$. If $u$ (respectively $v$) is not a vertex of $P$, then \Cref{lem:verticesAreGood} implies that the edge $\overline{u \,g}$ (respectively $\overline{v \,g}$) contains a reflex vertex of $P$ in its interior, which must appear as an interior vertex on the shortest path from $u$ to $v$.

    If only $v$ is maximal, then it again cannot be that $g = v$, but it may be that $g = u$. In this case the shortest path from $u$ to $v$ is the segment $\overline{u\,v}$, which is also left-turning. If $v$ is not a vertex of $P$, \Cref{lem:verticesAreGood} still implies there is at least one vertex on this shortest path.
\end{proof}

\begin{figure}[b]

    \centering
    \includegraphics[width = \linewidth]{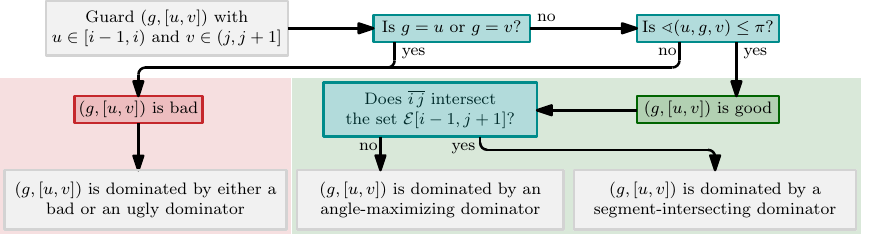}
    \caption{Flowchart illustrating how a guard is dominated. We prove the green block in~\Cref{lem:singleDominator} and the red block in~\Cref{lem:dominated_by_baddom}. }
    \label{fig:flowchart}
\end{figure}

\mysubpara{The good dominators.}
We construct a set $\reducedDom$ of $O(n)$ guards, called the \emph{reduced good dominators}, such that every good guard is dominated by one of them. These need not be good guards themselves! We start by defining two categories of good dominators, each of size $O(n^2)$.

\begin{observation}
For any fixed $g \in P$ and indices $i$ and $j$, with $i \leq j$, there exists at most one maximal chain $[u_{\text{max}}, v_{\text{max}}]$ visible from $g$ that contains $[i, j]$.
\end{observation}

\begin{proof}
    For any fixed $g \in P$, the points in $P$ visible from $g$ form a \emph{visibility polygon} $P_g$.
    Then $P_g \cap \partial P$ induces a fixed set of disjoint maximal chains which yields the observation.
\end{proof}

\begin{definition}
    \label{def:dominators}
    Consider an index pair $(i, j)$ with $i \in [2,n+1]$ and $j\in [i,2n]$.
    We define for $(i, j)$ two types of guards that we call \emph{good dominators}. Specifically, a segment-intersecting dominator~(\ref{itm:B}) and angle-maximizing dominator (\ref{itm:C}):
    \begin{enumerate}[label=\(\mathbf{\Alph*}\),ref=\(\mathbf{\Alph*}\), nolistsep, noitemsep]
        \item \label{itm:B} \textbf{(segment-intersecting):} See~\Cref{fig:segmentDom}. If $\core{i-1}{j+1}$ intersects the interior of $\overline{i\,j}$, we define the guard $(g, [u_{\text{max}}, v_{\text{max}}])$, 
where $g$ is the last point of $\overline{i\,j}$ intersecting $\core{i-1}{j+1}$,  and 
$[u_{\text{max}}, v_{\text{max}}]  \subseteq [i-1, j+1]$ is the maximal chain containing $[i, j]$ that is visible to $g$.
        \item \textbf{(angle-maximizing):}\label{itm:C} See~\Cref{fig:angleDom}. If $\core{i-1}{j+1}$ lies  left of  $\overline{i\,j}$, we define the guard $(g, [u_{\text{max}}, v_{\text{max}}])$. We define $g$ as $i$ (or $j$) if $\core{i-1}{j+1}$ intersects $i$ (or $j$). Otherwise, $g$ is the point in $\core{i-1}{j+1}$ that maximizes the angle $\sphericalangle(i,g,j)$. We define 
$[u_{\text{max}}, v_{\text{max}}] \subseteq [i-1, j+1]$ as its maximal visible chain containing  $[i, j]$.
    \end{enumerate}
    For $i = j$, we define $(g, [u_{\text{max}}, v_{\text{max}}])=(i,[i-1,i+1])$, which is a dominator of both type \ref{itm:B} and \ref{itm:C}. If there does not exist a type \ref{itm:B} and \ref{itm:C} guard for $(i,j)$ that adheres to the stated conditions, then there is no good dominator for $(i,j)$.
\end{definition}

Note that by general position if $\core{i-1}{j+1}$ is left of $\overline{i\,j}$, then it can intersect $\overline{i\,j}$ in at most a single point, so the angle-maximizing dominators are well-defined.

\begin{figure}

\begin{minipage}{0.48\textwidth}
  \centering
    \includegraphics[page=1]{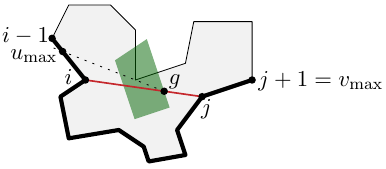}
    \caption{The visibility core $\core{i-1}{j+1}$ in green and the segment-intersecting dominator $(g,[u_\text{max},v_\text{max}])$ for $(i,j)$.}
    \label{fig:segmentDom}
    \end{minipage}
    \hfill
\begin{minipage}{0.48\textwidth}
    \centering
    \includegraphics[page=2]{pictures/goodDominators.pdf}
    \caption{The visibility core $\core{i-1}{j+1}$ in green and the angle-maximizing dominator $(g,[u_\text{max},v_\text{max}]) $ for $(i,j)$.}
    \label{fig:angleDom}
\end{minipage}
\end{figure}

\begin{definition}
The \emph{reduced good dominators} $\reducedDom$ are the good dominators not strictly dominated by another good dominator. 
\end{definition}

Next, we show that any good guard is dominated by a reduced good dominator. Our proof follows the approach illustrated in the green block of the flowchart in \Cref{fig:flowchart}.

\begin{lemma}
    \label{lem:singleDominator}
    For every good guard $(g, [u, v])$ there exists a reduced good dominator $(g', [u', v'])$ in $\reducedDom$ that dominates it.
\end{lemma}

\begin{proof}
    The guard $(g, [u, v])$ is good and thus $g \neq u,v$ and $\sphericalangle(v,g,u) \leq \pi$ by \Cref{def:dominators}. We prove that a good dominator (and thus a reduced good dominator) dominates $(g, [u, v])$. 
    Let $u \in [i-1, i)$ and $v \in (j, j+1]$. If $i = j$ or $i-1 = j$, then the dominator for $(i,j)$, which is the guard $(i,[i-1,j+1])$, dominates $(g,[u,v])$. So, assume $i < j$.
    
    We  define by $P_g$ the polygon formed by the edges $\overline{v\,g}$, $\overline{g\,u}$ and $[u, v]$. Since $(\overline{u \, g}, \overline{v \, g})$ forms a convex wedge, and $P_g$ contains no points of $\partial P$ in its interior,  all points in $\core{u}{v} \cap P_g$ see all of $[u, v]$. By general position, $\core{u}{v}$ either is left of $\overline{i\,j}$ and intersects $\overline{i\,j}$ in at most one point, or contains at least one point strictly right of $\overline{i\,j}$. We define $g^*$ as $i$ (or $j$) if $\core{u}{v}$ contains $i$ (or $j$) and is left of $\overline{i\,j}$, or as
    the point in $\core{u}{v}$ maximizing $\sphericalangle(i,g^*,j)$, otherwise.
    
\begin{figure}[b]

    \centering
    \includegraphics{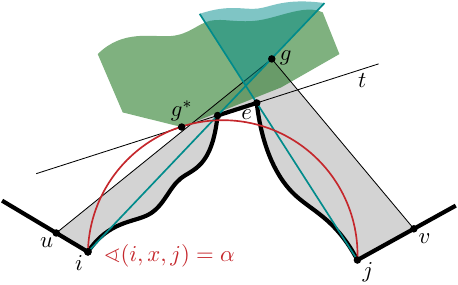}
    \caption{Case 2 of \Cref{lem:singleDominator}. The point $g^*$ that has $\sphericalangle(i,g^*,j)  = \alpha$ is in the visibility core $\core{u}{v}$ (green), but not in $P_g$ (gray). \vspace{-0.5cm}}
    \label{fig:dominator}
\end{figure}

   \textbf{\boldmath Case 1: $\sphericalangle(i, g^*, j) > \pi$.} 
The point $g^*$ must lie strictly right of the edge $\overline{i\,j}$ and, in particular, cannot serve as an angle-maximising dominator. 
Instead, we construct the segment-intersecting dominator (\ref{itm:B}) and show that it dominates $(g, [u, v])$. 
%By Condition~I or~II, either $(u, g, v)$ are collinear or $\sphericalangle(u, g, v) \leq \pi$. 
Because $\sphericalangle(u, g, v) \leq \pi$, we also have $\sphericalangle(i, g, j) \leq \pi$, which, together with $\sphericalangle(i, g^*, j) > \pi$, implies that $\overline{i\,j}$ lies in $P_g$. 
In particular, $g^*$ also lies in $P_g$, with $\overline{i\,j}$ separating $g$ and $g^*$. 
By definition, both $g$ and $g^*$ lie in $\core{u}{v}$. 
Since $\core{u}{v}$ is convex, and $\overline{i \, j}$ separates $g$ and $g^*$, $\core{u}{v}$ must intersect $\overline{i\,j}$. 
Moreover, every edge of $\core{u}{v}$ supports an edge in $[i{-}1, j{+}1]$, and as no three vertices of $P$ are collinear, there must exist an intersection point in the interior of $\overline{i \, j}$. 
Any intersection between $\overline{i \, j}$  and $\core{u}{v}$ lies in both $\core{u}{v}$ and $P_g$, and thus sees all of $[u, v]$. It follows that the last point of intersection $g'$ between $\overline{i \, j}$  and $\core{u}{v}$  sees $[u, v]$ and so the segment-intersecting dominator $(g', [u_{\max}, v_{\max}])$ corresponding to $(i, j)$ is well-defined and it dominates $(g, [u, v])$.

\textbf{\boldmath Case 2: $g^* = i$ or $g^* = j$ or $\sphericalangle(i,g^*,j)\leq \pi$.}
    If we can show that $g^*$ is in $P_g$, then it lies in $P_g \cap \core{u}{v}$ and thus sees all of $[u, v]$. So, the angle-maximizing dominator $(g^*, [u_{\max}, v_{\max}])$ (\ref{itm:C}) corresponding to $(i, j)$ dominates $(g, [u, v])$. If $g^* = i$ or $g^* = j$ then $g^*$ in $P_g$, concluding the proof.
    So, assume that $g^* \neq i,j$ and
    for the sake of contradiction that $g^*\not\in P_g$, refer to \Cref{fig:dominator}. Let $\alpha=\sphericalangle(i,g^*,j)$. The points $x\in\mathbb{R}^2$ such that $\sphericalangle(i,x,j)=\alpha$ form a circular arc passing through $i$ and $j$.
    As no other point in $\core{u}{v}$ has larger angle than $\alpha$, and $\core{u}{v}$ is convex, the tangent $t$ of this circular arc at $g^*$ separates the plane into two half-planes: a closed half-plane $H_{g}$ that contains $\core{u}{v}$ and an open half-plane $H_{i, j}$ containing the points $i$ and  $j$ (if $i$ or $j$ lies on $t$ then they coincide with $g^*$, a contradiction). 
    
    Now, consider the edges of $\core{u}{v}$ that intersect $g^*$. As $\core{u}{v}$ is contained in $H_g$, and $g \in \core{u}{v}$, the supporting line $\ell$ of one of these edges must pass between $g$ and $t$. As $P_g$ is contained in the wedge defined by $g,u,v$, the line $\ell$ intersects $P_g$ only in $H_g$.
    Note that $g^* \not \in P_g$ implies that $g^*$ lies strictly left of $\overline{u \, v}$ and so at least one of $u$ or $v$ is in $H_{i,j}$. 
    
    Let $e \in E$ be the edge defining $\ell$. By definition, $e$ lies on $\ell$ and thus $\ell$ separates $e$ from $i$ and $j$. 
    Consider the directed line $\ell_j$ from the right endpoint of $e$ to $j$, and consider the chain from $e$'s right endpoint to $j$.
    This chain must contain at least one edge $e^*$ whose start lies left of $\ell_j$ and whose endpoint right right of $\ell_j$ (recall that we distinguish between right and \emph{strictly} right). It follows that $\core{u}{v}$ lies left of $\ell_j$.
    Defining $\ell_i$ analogously yields a wedge bounded by $\ell_i$ and $\ell_j$ that contains $\core{u}{v}$ (see the blue region in \Cref{fig:dominator}).
    The fact that one of $u$ or $v$ is in $H_{i, j}$ implies that the point of intersection between $\ell_i$ and $\ell_j$ lies in $H_g$ and not on $t$. Moreover, $g$ also lies in $H_g$ and not on $t$, thus no point on $t$ is in $\core{u}{v}$, which contradicts $g^*$ being in $\core{u}{v}$.
\end{proof}

\subsection{Bad and ugly dominators}

In addition to good dominators, we also define bad and ugly dominators. 
Bad dominators are defined for each index pair $(i, j)$ in an analogous manner to good dominators:

\newcommand{\baddomref}{\hyperref[def:bad-dominator]{\textnormal{(\(\badSymbol\))}}\xspace}

\begin{definition}[Bad dominators \baddomref{}]
\label{def:bad-dominator}
     See Figure~\ref{fig:badDom}. Consider an index pair $(i, j)$ with $i \in [2,n+1]$ and $j\in [i,2n]$.
     For every vertex $c$ of $\core{i-1}{j+1}$, we define at most one bad dominator.  Each vertex $c$  of $\core{i-1}{j+1}$ is defined by edges $[a-1, a]$ and $[b, b+1]$. 
     Let $[u_\text{max}, v_\text{max}]$ be the maximal chain containing $[a-1, b+1]$ that is visible from $c$. 
     \begin{itemize}[nolistsep]
         \item        If $[u_\text{max}, v_\text{max}]$ exists and is non-empty, then we define for $c$ the bad dominator  $(c, [u_\text{max}, v_\text{max}])$.
     \end{itemize}
\end{definition}

\begin{figure}[b]

\begin{minipage}{0.48\textwidth}
  \centering
    \includegraphics[page=1]{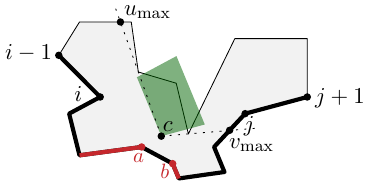}
    \caption{The visibility core $\core{i-1}{j+1}$ in green and the bad dominator $(c,[u_\text{max},v_\text{max}])$ for the vertex $c$ defined by the red edges.}
    \label{fig:badDom}
    \end{minipage}
    \hfill
\begin{minipage}{0.48\textwidth}
    \centering
    \includegraphics[page=2]{pictures/badDominators.pdf}
    \caption{The visibility core $\core{i-1}{j+1}$ in green and the ugly dominator $(g^*,[u,v^*]) $ for the bad guard $(g,[u,v])$ with path $S$ in red.}
    \label{fig:uglyDom}
\end{minipage}
\end{figure}

\begin{definition}\label{def:reduced_bad_dom}
    The \emph{reduced bad dominators} \baddom are the bad dominators that are not strictly dominated by any other guard. 
\end{definition}

\newcommand{\uglydomref}{\hyperref[def:ugly-dominator]{\textnormal{(\uglySymbol)}}}

\noindent
Finally, the ugly dominators are  defined off of an existing bad guard $(g, [u, v])$:
\begin{definition}[Ugly dominators \uglydomref{}] \label{def:ugly-dominator}
See Figure~\ref{fig:uglyDom}. Let $(g, [u, v])$ be a bad guard with $v = \nex{u}$ and  $u \in [i-1, i)$ and $v \in (j, j+1]$. 
         \label{type:2} Let $S$ be the shortest path from $u$ to $j+1$ in $P$, and let $\ell$ be the supporting line of the first edge. 
        If $\ell$ intersects $\core{i-1}{j+1}$, let $g^*$ be the last intersection point along $\ell$, and let $v^*$ be the farthest point on $[j, j+1]$ visible from $g^*$.  
        \begin{itemize}[nolistsep]
            \item      If $g^*$ and $v^*$ exist then we define the \emph{ugly dominator} $(g^*, [u, v^*])$. 
        \end{itemize}
\end{definition}

\begin{lemma} \label{lem:dominated_by_baddom}
For any bad guard $(g, [u, v])$ where $u\in[i,i+1)$ and $v = \nex{u}\in (j,j+1]$, there exists either a bad dominator \baddomref{} or an ugly dominator \uglydomref{} that dominates $(g, [u, v])$.
\end{lemma}

\begin{figure}[b]

\begin{minipage}{0.48\textwidth}
  \centering
    \includegraphics{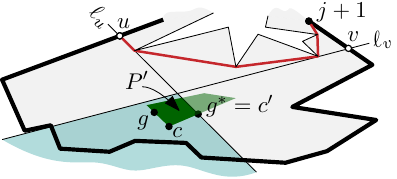}
    \caption{The shortest path from $u$ to $j+1$ (red), and $\core{i-1}{j+1}$ (green). Any point that sees both $u$ and $v$ lies in the blue wedge.}
    \label{fig:guard}
    \end{minipage}
    \hfill
\begin{minipage}{0.48\textwidth}
    \centering
    \includegraphics{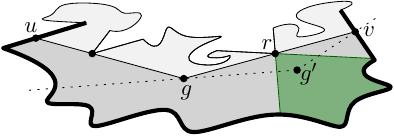}
    \caption{Illustration of \Cref{thm:good_guards} which shows a bad contiguous guard $(g, [u, v])$.\\
    \textcolor{white}{hidden}}
    \label{fig:oneGoodGuard}
\end{minipage}
\end{figure}

\begin{proof}
    Let now $S_v$ be the shortest path from $u$ to $v$ and denote by $\ell_u$ the directed supporting line of its first edge. By \Cref{cor:left-turning}, $S_v$ forms a left-turning chain. 
    Let $S$ be the shortest path from $u$ to $j+1$. 
      Since $(g, [u, v])$ has maximal $v$, \Cref{lem:verticesAreGood} states that $\overline{g \, v}$ contains a reflex vertex of $P$. If $v \neq j+1$, $\overline{g \, v}$ contains a reflex vertex in its interior. This implies the first edge of $S_v$ and $S$ are the same, and thus the line $\ell_u$ equals $\ell$.
    If $v = j+1$, then $S_v$ is equal to $S$, so also $\ell$ equals $\ell_u$. 
    Denote by $P'$ the convex polygon that is $\core{i-1}{j+1}$ intersected with the half-plane right of $\ell_u$, see \Cref{fig:guard}. 
    Denote by $c$ a point in $P'$ that sees farthest along $[j, j+1]$. We can choose $c$ to be a vertex of $P'$.
    We make a case distinction:
    \begin{enumerate}[label=(\roman*),topsep=4pt, partopsep=0pt, parsep=0pt, itemsep=0pt]
    \item If $c$ is also a vertex of $\core{i-1}{j+1}$ then $c$ defines a bad dominator \baddomref{} which is the guard $(c, [u_\text{max}, v_\text{max}])$. By maximality, $[u, v]$ must be contained in $[u_\text{max}, v_\text{max}]$. 
    \item if $c$ is not a vertex of  $\core{i-1}{j+1}$  then $c$ must be an intersection point of $\ell_u = \ell$ and $\core{i-1}{j+1}$. 
    As $\ell$ and $\core{i-1}{j+1}$ intersect, the ugly dominator $(g^*, [u, v^*])$ exists, where $g^*$ is the last point of intersection between $\ell$ and $\core{i-1}{j+1}$. 
    As the chain $S_v$ is a left-turning chain, the intersection point of $\ell$ and $\core{i-1}{j+1}$ that sees furthest along the edge $[j,j+1]$ is the last point of intersection, i.e. $g^* = c'$. As $c'$ sees at least as far as $g$, because $g \in P'$, it must be that $v^* \geq v$.
     \qedhere
    \end{enumerate}
\end{proof}

\subsection{The \nexFunc function and our high-level algorithm}

Define by $X$ a set of points along $\partial P$ that contains all vertices of $P$ and for all guards $(g, [u, v]) \in \reducedDom$ the point $u$. 
We show that we can compute an optimal solution to the \contArt problem by starting at a guard in the set of reduced good dominators $\reducedDom$ or at a vertex of $P$, and greedily applying the `next best' guard.

\begin{definition}[The $\nexFunc$ function \cite{theEnemy, RobsonSpaldingZheng2024_AnalyticArcCover}]
For $u \in [1, n+1]$, $\nex{u}$ returns the largest $v \in [u, 2n+1]$ where there exists a guard $(g, [u, v])$.
\end{definition}

\noindent

\begin{theorem}
    \label{thm:good_guards}
    Let $G$ be a minimal set of contiguous guards guarding $\partial P$ with $|G| > 1$. There exists a set of contiguous guards $G'$ with $|G| = |G'|$ where at least one $(g, [u, v]) \in G'$ is either in $\reducedDom$, or has $u$ a vertex of $P$.
\end{theorem}

\begin{proof}
    We show how to transform $G$ into such a set $G'$ (see also \Cref{fig:oneGoodGuard}).
    We start by replacing every guard $(g, [u, v]) \in G$ by a guard $(g, [u', v'])$ such that $[u, v] \subseteq [u', v']$ and the interval $[u', v']$ is maximal to form $G'$. Next, we show that we can adapt $G'$ such that either there is a good guard in $G'$ or a guard whose visibility ends at a vertex of $P$. \Cref{lem:singleDominator} implies that this good guard is dominated by a guard in $\reducedDom$, concluding the proof.
    
    Suppose that $G'$ contains only bad guards and let $(g, [u, v]) \in G'$ be such a bad guard. 
    If $(u, v)$ does not contain a vertex of $P$ then $[u, v] \subseteq [i, i+1]$ for some vertex $i$.
    We now replace $(g, [u, v])$ by $(i, [i, i+1])$ and conclude the proof.
    Similarly, if $v$ is a vertex of $P$ then, because $|G| > 1$, there has to be at least one other guard $(g', [u', v']) \in G'$ such that $v \in [u', v']$. We replace this guard by $(g', [v, v'])$ so that the start of its chain is a vertex of $P$ which concludes the proof. 
    Lastly, note that it is not possible for $g$ to be $u$ (resp. $v$), as in that case, $g$ sees all edges of $P$ containing $g=u$ (resp. $g=v$), and in particular, strictly more than~$[u,v]$, contradicting the maximality of~$[u,v]$.

It follows that for all bad guards $(g, [u, v]) \in G'$, the open chain $(u, v)$ contains a vertex of $P$, while neither $u$ nor $v$ is a vertex of $P$. 
By \Cref{lem:verticesAreGood}, the segment $\overline{v\,g}$ contains a reflex vertex $r$ of $P$ with $r \notin [u, v]$. 
Hence, there exists a guard $(g', [u', v'])$ such that $r \in [u', v']$. 
If $u' = r$ or $v' = r$, we conclude the proof as above. 
Otherwise, if $r \in (u', v')$, then $g'$ lies left of both supporting lines of the edges incident to $r$. 
This implies that $g'$ lies left of $\overline{v\,g}$. 
Since $\sphericalangle(v,g,u) > \pi$, any point that sees $r$ and lies left of $\overline{v\,g}$ must also lie within the polygon $P_g$ defined by $[u, v]$, $\overline{v\,g}$, and $\overline{g\,u}$ (see \Cref{fig:oneGoodGuard}). 
Observe that the points in $[u', v'] \setminus (u, v)$ form a single closed chain; otherwise, $[u, v] \subseteq [u', v']$, contradicting the minimality of $G'$. 
Define $[a, b] = [u', v'] \setminus (u, v)$. 
If we replace $(g', [u', v'])$ by $(g', [a, b])$, the resulting set of guards still covers $\partial P$. 
Crucially, $[a, b]$ lies left of $\overline{g'\,v}$ and left of either the supporting line of $\overline{u\,g'}$ or $\overline{g\,g'}$. 
As $g'$ lies in $P_g$ and $\sphericalangle(v,g,u) > \pi$, this implies that, unless $g'=a$ or $g'=b$, $\sphericalangle(a,g',b) \leq \pi$, and thus $G'$ contains a good guard. Finally, if $g'=a$ (resp. $g'=b$), then $g'$ is on $\partial P$ and in particular, sees the entirety of the edge containing $a$ (resp. $b$), hence we can extend $[a,b]$ to start (or end) in a vertex, which concludes the proof.
\end{proof}

\noindent
The following theorem gives the high-level idea of the running times achieved in \Cref{sec:convexHullDataStructures,sec:computingDiscreteDominators,sec:knlogn}.

\begin{theorem}
    \label{thm:basic_runtime}
    Let $P$ be a simple polygon of $n$ vertices and $X$ be a linear-size set that includes all polygon vertices and for all $(g, [u, v]) \in \reducedDom$, the point $u$. 
    Given $P$, $X$ and an $O(T)$-time implementation of the $\nexFunc$ function, we can compute an optimal solution of size $k$ to the \contArt problem in $O( k n \, T)$ time. 
\end{theorem}

\begin{proof}
We apply to each $x \in X$ the function $\nex{x}$ recursively until a set of guards covering $\partial P$ is obtained. 
Each such sequence contains at most $k+1$ guards and this thus takes $O(k \, T  n)$ time.

    We claim that one of these guard sequences must have size $k$. 
    By \Cref{thm:good_guards} there exists at least one guarding solution of size $k$ that contains a guard $(g, [x, v])$ where the guard is either in $\reducedDom$ or $x$ is a vertex of $P$.
    A classical argument now implies that recursively applying \nex{x} finds a solution of size $k$ also. 
    For completeness, we write the argument here: 
    Let $G = \{ (g_i, [v_{i-1}, v_{i}])\}$ be minimum ordered sequence of guards that guards $\partial P$ that is generated by recursively applying $\nex{x}$. 
    Then $(g_1, [x_1, x_2]) = (g, [x, \nex{x}])$. 
    
    Suppose, for the sake of contradiction, that $|G| > k$. 
    Let $G^*$ be an optimal solution of size $k$ that includes $(g, [x, v])$ and order this solution along $\partial P$, starting at $(g, [x, v])$. 
    Since $|G| > |G^*|$,  there must exist a minimum index $i$ such that for the $i$\ts{th} guards $(g_i, [v_{i-1}, v_{i}])\in G$ and $(g_i^*, [u_i^*, v_i^*])\in G^*$, it holds that $v_{i} < v_i^*$. 
    However, since $i$ is minimum, it must be that $v_{i-1} \geq u_i^*$. 
    But then $v_{i-1} \in [u_i^*, v_i^*]$ which makes $\nex{v_{i-1}}$ return a value that is at least $v_i^*$---a contradiction.
\end{proof}

%\newpage
\section{Intermezzo: data structures and sliding sequences}\label{sec:convexHullDataStructures}

Our results rely on several classical geometric results concerning simple polygons.

\begin{datastructure}[\cite{Guibas1987Optimal}]\label{ds:shortestPath}
A size-$n$ polygon $P$ can be stored in $O(n)$ time and space so that, for any $s, t \in P$, the shortest path $\short{s}{t}$ can be reported as $O(\log n)$ balanced trees in $O(\log n)$ time. Moreover:
\begin{itemize}[noitemsep, nolistsep]
\item $\short{s}{t}$ is represented by $O(\log n)$ balanced trees whose ordered leaves traverse $\short{s}{t}$.
    \item We can find the first edge $e$ along $\short{s}{t}$ where the path makes a left turn around the first vertex of $e$ and a right turn around the last vertex of $e$, in $O(\log n)$ time. 
\end{itemize}
\end{datastructure}

\begin{proof}
The data structure of~\cite{Guibas1987Optimal} stores a family of shortest paths in $P$ as balanced binary trees. 
For any query points $s, t \in P$, the shortest path $\short{s}{t}$ can be expressed as $O(\log n)$ subtrees, together with $O(\log n)$ newly computed edges called \emph{bridges}. 

The second property follows from a standard adaptation of this data structure (see also~\cite{Eades2020Visibility}): 
for each pair of consecutive edges on a precomputed shortest path, determine whether the turn is left or right. 
For each node in the tree, store a symbol indicating whether all of its descendant edges are left-turning, right-turning, or mixed. 
For a query, we can determine for each of the $O(\log n)$ newly computed bridges whether the turn at its endpoints is left or right. 
Combined with the pre-stored symbols, 
this information suffices to answer the second query.
\end{proof}

\begin{datastructure}[\cite{Hershberger1993Ray}]\label{ds:rayShooting}
A size-$n$ polygon $P$  can be stored in $O(n)$ time and space so that, for any ray $r$ whose origin lies in $P$, the point in $\partial P$ hit by $r$ can be found in $O(\log n)$ time.
\end{datastructure}

\noindent
We always assume that we have access to \Cref{ds:shortestPath} and \Cref{ds:rayShooting} for the given polygon.
We repeatedly use the concept of a \emph{sliding sequence}. 
A sliding sequence can be thought of as a pair of indices $(i, j)$ that can be updated by incrementing either $i$ or $j$. 
Formally, \Cref{def:sequence} specifies a conforming sliding sequence used by almost all of our algorithms. 
For each $(i, j)$, we maintain the visibility core $\core{i-1}{j+1}$ as a convex-hull data structure supporting emptiness, containment, extreme-point, and ray-shooting queries, and an additional \emph{angle query} (defined later). 
For any linear-size sliding sequence $\{ (i_t, j_t) \}$, Chan, Hershberger, and Pratt~\cite{Chan2019Time-Windowed} give a data structure supporting such updates and queries in $O(\log n)$ time:

\begin{lemma}[Paraphrased Lemma 1 in \cite{Chan2019Time-Windowed}] \label{lem:core_total_size}
    %For any visibility core $\core{i-1}{j+1}$, denote by $V(\core{i-1}{j+1})$ its vertices.
    For any linear-size sliding sequence $\sigma$, the total number of distinct vertices that get added or removed from the visibility core $\core{i-1}{j+1}$ when iterating over $(i, j) \in \sigma$ is linear.
\end{lemma}

\begin{lemma}\label{lem:contiguousInterval}
    During a traversal of any sliding sequence $\sigma$ each point $p$ appears as a vertex on the visibility core in at most one contiguous time interval.
\end{lemma}

\begin{proof}
   Let $v$ be a vertex of the visibility core $\core{i-1}{j+1}$ that is not a vertex of $\core{i'-1}{j'+1}$, where $(i', j')$ follows $(i, j)$ in $\sigma$. 
    There are two possible reasons for $v$ not being a vertex of $\core{i'-1}{j'+1}$: either one of its supporting edges has been removed from the sliding window, or a new half-plane has been added to the visibility core that excludes $v$. 
    In the first case, since removed edges never reappear in later windows, $v$ cannot return. 
   In the second case, the new half-plane remains in the sliding window at least as long as the edges defining $v$, again preventing $v$ from reappearing. 
    Hence, each vertex appears in at most one contiguous interval.
    \end{proof}

The above two lemmas imply the following, see also the remarks made in \cite{Chan2019Time-Windowed}.
\begin{corollary}\label{thm:fifo}
        One can maintain the intersection of half-planes, explicitly, stored as a balanced binary tree with leaves in cyclical order, under first-in-first-out updates using linear space and an amortized update time $O(\log n)$. 
\end{corollary}

We subsequently implement a very specific data structure that we use throughout the paper:

\begin{datastructure}\label{ds:sliding}
    We define \slidingds as a data structure that for a pair of indices $(i, j)$ stores the visibility core $\core{i-1}{j+1}$ and $\core{i-1}{j+2}$, supporting the following operations in (amortized) $O(\log n)$ time:
    \begin{itemize}[noitemsep, nolistsep]
        \item \textbf{\boldmath Increment $i$:} Delete the corresponding half-plane.
        \item \textbf{\boldmath Increment $j$:} Insert a new half-plane. 
        \item \textbf{Emptiness query:} Determine whether $\core{i-1}{j+1}$ is empty and return an arbitrary point inside it if not.
        \item \textbf{Containment query:} For a point $x$, determine whether $x\in\core{i-1}{j+1}$.
        \item \textbf{Extreme-point query:} For a given vector $w$, return the $x \in \core{i-1}{j+1}$ maximizing  the dot product $\langle w,x\rangle$.
        \item \textbf{Ray-shooting query:} Return the last point of intersection between a ray $\ell$ (a directed half-line) and $\core{i-1}{j+1}$.
        \item \textbf{Angle query:} Given $u,v$ with $\core{i-1}{j+1}$ left of $\overline{u\,v}$, find $g \in \core{i-1}{j+1}$ that maximizes the angle $\sphericalangle(u, g, v)$.
        \item We also support all these queries for $\core{i-1}{j+2}$. Whenever we perform a query on $\core{i-1}{j+2}$ instead of $\core{i-1}{j+1}$, we refer to such a query as a \emph{lookahead query}.
    \end{itemize}
\end{datastructure}

\begin{restatable}{theorem}{slidingdsthm}
    The \slidingds can be implemented with $O(n)$ preprocessing and space, supports queries in $O(\log n)$ time, and updates in amortized $O(\log n)$ time. 
\end{restatable}

\begin{proof}
Let $(i, j)$ be the dynamic index pair from the \slidingds. We maintain $\core{i-1}{j+1}$ and $\core{i-1}{j+2}$ separately, under increments of $i$ and $j$ using \Cref{thm:fifo}. This maintains the vertices of these convex areas as balanced binary trees where the left-to-right traversal of the leaves corresponds to a clockwise traversal. 
For a convex polygon with $n$ vertices, whose vertices are stored in a balanced tree according to their cyclic order, the containment, ray shooting, and extreme-point queries can be answered in $O(\log n)$ time~\cite{DOBKIN1983}.
We can thus immediately support all queries except for angle queries.

For angle queries, we conceptually rotate and translate the plane until $\overline{u \, v}$ coincides with the horizontal line through $(0, 0)$ with the area left of $\overline{u \, v}$ lying above this line. 
For any fixed angle $\gamma > 0$, the points $\vec{x} \in \mathbb{R}^2_{\geq 0}$ for which $\sphericalangle (u, \vec{x}, v) =  \gamma$ form a convex arc.
For $\gamma_1, \gamma_2$ with $\gamma_1 > \gamma_2$ the corresponding convex arc lies strictly above the arc induced by $\gamma_2$. It follows that the function $\vec{x} \mapsto-\sphericalangle(u,\vec{x},v)$ defined for all points $\vec{x} \in \mathbb{R}^2_{\geq 0}$ is well-defined, unimodal in $\core{i-1}{j+1}$, and strictly monotone between adjacent edges of $\core{i-1}{j+1}$. 
Since we have the edges of $\core{i-1}{j+1}$ explicitly, in-order, in a balanced binary tree, it then follows that the point in $g$ that maximizes $\sphericalangle(i, g, j)$  lies on the boundary of  $\core{i-1}{j+1}$ and that we can find it through binary search along the edges of $\core{i-1}{j+1}$  in $O(\log n)$ time. 
\end{proof}

\section{Computing all reduced  dominators}\label{sec:computingDiscreteDominators}

We now give an immediate application of \Cref{obs:trivialOneCenter} and our data structures, which we will frequently use to compute the reduced dominators: 

\begin{lemma}\label{lem:computeVisibility}
Let $g \in P$ be a fixed point, and let $s, m, t \in \partial P$ with $m\in[s,t]$. 
Suppose $g$ lies in $\core{s}{t}$. 
Then there exists a unique maximal chain $[u, v] \subseteq [s, t]$ that is visible from $g$ and that contains $m$, which can be found in $O(\log n)$ time.
\end{lemma}

\begin{proof}
If $g$ does not see $m$, we detect this in $O(\log n)$ time using \Cref{ds:shortestPath} and return the empty interval. 
Let $S_s$ be the shortest path from $g$ to $s$, and let $r_s$ denote the ray along its first edge. 
Define $S_t$ and $r_t$ analogously. 
The three elements $(r_s, g, r_t)$ form a wedge. 
We claim that the points $u$ and $v$ of $\partial P$ that are respectively hit by $r_s$ and $r_t$ define the desired maximal chain. 
We can find these points in $O(\log n)$ time using two ray-shooting queries (\Cref{ds:rayShooting}).

\begin{figure}[b]

    \centering
    \includegraphics{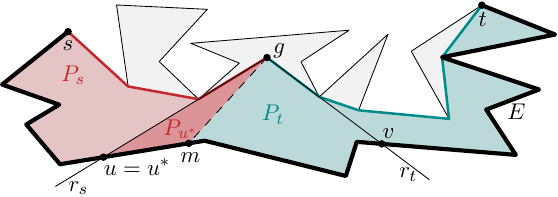}
    \caption{The inclusion-wise maximal area visible from $g$ and containing $m$ is bounded by the rays $r_s$ and $r_t$ along the first edges of the shortest paths from $g$ to $s$ and $t$, respectively. 
    These rays hit $\partial P$ in points $u$ and $v$ that define the maximal chain visible to $g$.}
    \label{fig:visiblityCompute}
\end{figure}

Consider the polygon $P'$ bounded by $S_s$, $[s, t]$, and $S_t$. 
Since $g$ sees $m$, the segment $\overline{g\,m}$ splits $P'$ into two polygons, $P_s \subseteq P$ and $P_t \subseteq P$, containing $s$ and $t$ respectively (see \Cref{fig:visiblityCompute}). 
Let $u^*$ be the maximal point in $[s, m]$ such that $\overline{g\,u^*}$ intersects $S_s$ in more than one point (i.e., not only in $g$). 
We first show that $g$ sees all of $[u^*, m]$. 
Indeed, consider the polygon $P_{u^*}$ bounded by $[u^*, m]$, $\overline{m\,g}$, and $\overline{g\,u^*}$. 
Because $g \in \core{s}{t}$, it lies left of all edges of $P_{u^*}$ and thus, by \Cref{obs:trivialOneCenter}, the visibility from $g$ to $[u^*, m]$ is not blocked by any edge in $[s,t]$.
By construction, $P_{u^*}$ contains no point of $S_s$ in its interior and hence lies within $P_s \subseteq P$. 
Therefore, $g$ indeed sees $[u^*, m]$ in $P$. 

If $s \neq u^*$, then $g$ does not see any point $u' \in [s, u^*)$, since by definition of $u^*$, the segment $\overline{g\,u'}$ intersects an edge of $S_s$ in its interior. 
It remains to show that $u^* = u$. 
Let $r^*$ be the ray from $g$ through $u^*$. 
We claim that $r^* = r_s$. 
Indeed, since $\overline{g\,u^*}$ intersects $S_s$, if this intersection occurs at a point $r$ not lying on the last edge of $S_s$, then $S_s$ could be shortened by the segment $\overline{r\,g} \subseteq \overline{u^*\,g} \subseteq P$. 
This contradicts the definition of $S_s$ as a shortest path. 
Hence $r^* = r_s$, and we conclude that $u^* = u$. 
By symmetry, the same argument applies to $r_t$, yielding the maximal visible chain $[u, v]$.
\end{proof}

We compute two linear-size sets $\reducedDom'$ and $\baddomsigma$ that contain all reduced good and reduced bad dominators, respectively. 

\subsection{A linear-size set $\reducedDom'$ that contains all reduced good dominators in  $\reducedDom$}

We compute a linear-size set $\reducedDom'$ that contains all guards in the set of reduced good dominators $\reducedDom$ (\Cref{def:reduced_bad_dom}) via a sliding window.
Formally, we maintain a pair of indices $(i, j)$ subject to incrementing $i$ and incrementing $j$.
During this procedure, we maintain $\core{i-1}{j+1}$ (and  $\core{i-1}{j+2}$) via \slidingds and we prove that this data structure can answer an advanced query:

\begin{lemma}\label{lem:finding_BUC}
    Given vertices $i$ and $j$ with $i \leq j$, and $\core{i - 1}{j + 1}$ in \slidingds, we can compute the unique good dominator for $(i, j)$ (see \Cref{def:dominators}), or return that no good dominator exists for $(i, j)$, in $O(\log n)$ time.
\end{lemma}

\begin{proof}
    We perform an emptiness query on $\core{i-1}{j+1}$.
If the visibility core $\core{i-1}{j+1}$ is empty then by definition no dominator exists for $(i, j)$. 

We first use extreme-point queries in the two directions orthogonal to $\overline{i\,j}$ to test for points in $\core{i-1}{j+1}$ strictly left and right of $\overline{i\,j}$. 
If there is no point strictly right of $\overline{i\,j}$ then the entire visibility polygon lies left of $\overline{i\,j}$. In this case, we can, via angle queries and containment queries at $i$ and $j$, compute the point $g^*$ that is either $i$ or $j$, if $i$ or $j$ is in $\core{i-1}{j+1}$, and the point in $\core{i-1}{j+1}$ that maximizes the angle $\sphericalangle(i, g^*, j)$.
If instead the visibility core has both a point left of $\overline{i\,j}$ and a point strictly right of $\overline{i\,j}$, then we compute a point $g'$ that is in the interior of $\overline{i\, j}$ and in $\core{i-1}{j+1}$, via two ray shooting queries on $\core{i-1}{j+1}$ using the supporting ray from $i$ to $j$, and the ray from $j$ to $i$.
If there is no point left of $\overline{i\, j}$, or $\core{i-1}{j+1}$ has a point strictly right of $\overline{i\, j}$, but only intersects $\overline{i\, j}$ in either $i$ or $j$, then by definition there is no good dominator for $(i,j)$.
Note that at most one of $g^*$ or $g'$ exists.

We apply \Cref{lem:computeVisibility} to compute the maximal interval $[u_\text{max}, v_\text{max}]$ that is visible to $g^*$ (or $g'$) that includes $i$ in logarithmic time.
If $[i, j] \subseteq [u_\text{max}, v_\text{max}]$ then we return this guard as the segment-intersecting or angle-maximizing dominator for $i, j$. Otherwise, this procedure certifies that no such dominator exists. 
\end{proof}

\begin{lemma}
    \label{lem:no_further_dominator}
    Suppose that for an index pair $(i, j)$ with $i \in [2, n+1]$ and $j \in [i, 2n]$, there exists no good dominator. Then there exists no good dominator for $(i, j')$ with $j' > j$. 
\end{lemma}

\begin{proof}
For any pair $(i', j')$, define the \emph{potential dominator} $g$ as follows:  
If the interior of $\overline{i'\,j'}$ intersects $\core{i'-1}{j'+1}$, then $g$ is the last point of intersection between $\core{i'-1}{j'+1}$ and $\overline{i'\,j'}$.  
If $\core{i'-1}{j'+1}$ lies left of $\overline{i'\,j'}$ and intersects $i$ or $j$, then $g$ is equal to $i$ or $j$, respectively. If $\core{i'-1}{j'+1}$ lies strictly left of $\overline{i'\,j'}$, then $g$ is the point in the visibility core that maximizes $\sphericalangle(i', g, j')$.  
If none of these cases apply, then $g$ is undefined.  
The remainder of the proof proceeds by case distinction, depending on whether $g$ is undefined for $(i, j)$ or does not realize a good dominator.  
Observe that a good dominator for $(i, j)$ exists if and only if the potential dominator $g$ exists and its maximal visible chain $[u_{\text{max}}, v_{\text{max}}]$ satisfies $u_{\text{max}} \in [i-1, i]$ and $v_{\text{max}} \in [j, j+1]$.

    \textbf{Case 1:  $\core{i-1}{j+1}$ is empty.} Then $\core{i-1}{j'+1}$ is also empty for all $j' > j$.

    \textbf{Case 2: $\core{i-1}{j+1}$ is not empty but $g$ is undefined. }
    We claim that for all pairs $(i, j')$ with $j' \geq j$, either this property also holds or $\core{i-1}{j'+1} = \emptyset$.
    Since the vertices of $P$ lie in general position and $\core{i-1}{j+1}$ does not intersect the interior of $\overline{i\,j}$, the visibility core $\core{i-1}{j+1}$ must be contained in the polygon $P'$ defined by $\overline{i\,j}$ and $[i,j]$. 
    
    Consider any $j'> j$. 
    We first assume that $j'$ lies strictly right of $\overline{i\,j}$ and distinguish two subcases based on the position of $j'$.
    \Cref{fig:no_better_j_2}: Suppose that $j'$ does not lie in $P'$. Then no point in $\core{i-1}{j+1}$ can see $j'$ and so there exists no good dominator for $(i,j')$. \Cref{fig:no_better_j_3}: Suppose $j'$ is in $P'$. Consider the subset of $\core{i-1}{j+1}$ that is left of $\overline{i\,j'}$. Let $a$ be the last point of intersection of $[j,j']$ and $\overline{i\,j}$. For any point $p \in \core{i-1}{i+1}$ that lies in the polygon $B$ defined by $\overline{i\,j'}$, $[j',a]$, and $\overline{a\,i}$, it must be that $p \notin \core{i-1,j'+1}$. 
    So consider a point $p \in \core{i-1}{i+1}$ that does not lie in $B$ (not in the blue area in \Cref{fig:no_better_j_3}).
    If $p$ lies strictly left of $\overline{i\,j'}$ then $p$ cannot see $i$, thus $p$ cannot a the good dominator for $(i, j')$. 
    Furthermore, the only point of $\overline{i \, j'}$ that can lie in $\core{i-1}{j'+1}$ is the vertex $j'$ itself. It follows that no segment-intersection dominator (\ref{itm:B}) exists for $(i, j')$. If $\core{i-1}{j'+1}$ would be left of $\overline{i}{j'}$ then there would be an angle-maximizing dominator (\ref{itm:C}) for $(i,j')$. However, then either $[j'-1,j']$ is collinear with $i$, which contradicts our general position assumption, or $j'$ cannot see $i$. We conclude that no good dominator exists for $(i,j')$.
    
    Suppose otherwise that $j'$ lies left of $\overline{i\,j}$. Then because $P$ lies in general position, $j'$ must lie strictly left of this line. 
    However, then the interior of $\overline{i \, j'}$ lies outside of the polygon $P'$ defined by $\overline{i\,j}$ and $[i,j]$.  We observed that $\core{i-1}{j+1}$ is contained in $P'$ and since $\core{i-1}{j'+1} \subseteq \core{i-1}{j+1}$ it follows that no segment-intersecting dominator (\ref{itm:B}) exists for $(i, j')$.
    Finally if $j'$ lies left of $\overline{i \, j}$ and $\core{i-1}{j'+1}$ is strictly left of $\overline{i\,j'}$ then no point in $\core{i-1}{j'+1}$ can see $j'$ (see \Cref{fig:no_better_j}),  Thus, there also exists no angle-maximizing dominator (\ref{itm:C}) for $(i, j')$.

 \textbf{Case 3: the potential dominator $g$ exists for $(i, j)$, but its maximal visible chain $[u_{\text{max}}, v_{\text{max}}]$ satisfies $u_{\text{max}} \notin [i-1, i]$ or $v_{\text{max}} \notin [j, j+1]$.}  
We reduce this case to Case~2.  
Let $P^\uparrow$ denote the subpolygon of $\core{i-1}{j+1}$ left of $\overline{i\,j}$.  
Consider any point $p \in P^\uparrow$.  
If $p$ sees both $i$ and $j$, then $(p, [i, j])$ forms a good guard (it satisfies Condition~II).  
By \Cref{lem:singleDominator}, there then exists a good dominator for $(i, j)$ that dominates $(p, [i, j])$, contradicting the assumption of the lemma.  
Hence, no point in $P^\uparrow$ sees both $i$ and $j$.  
Since any potential dominator for $(i, j')$ with $j' > j$ must see at least $[i, j]$, we may regard the visibility core $\core{i-1}{j+1}$ as excluding $P^\uparrow$.  
This adjusted visibility core now lies strictly right of $\overline{i\,j}$ and thus has no potential dominator.  
Consequently, the situation reduces to Case~2.  \qedhere

\begin{figure}

\begin{minipage}{0.31\textwidth}
  \centering
    %\begin{figure}
            \centering
        \includegraphics[page=2]{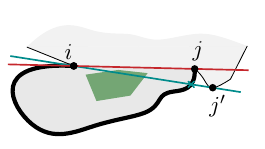}
        \caption{Any guard in the visibility core (green) will not see $j'$.}
        \label{fig:no_better_j_2}
    %\end{figure}
    \end{minipage}
    \hfill
\begin{minipage}{0.31\textwidth}
    %\begin{figure}
        \centering
        \includegraphics[page=5]{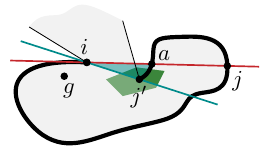}
        \caption{No point in the blue polygon can be in the visibility core $\core{i-1}{j'+1}$.}
        \label{fig:no_better_j_3}
    %\end{figure}
\end{minipage}
    \hfill
\begin{minipage}{0.31\textwidth}
  \centering
    %\begin{figure}
        \centering
        \includegraphics[page=4]{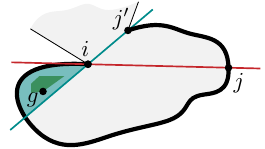}
        \caption{The potential angle-maximizing dominator for $(i,j')$ does not see $j'$.}
        \label{fig:no_better_j}
    %\end{figure}
    \end{minipage}
\end{figure}
\end{proof}

\begin{theorem}\label{lem:compute_B_C}
    For a simple polygon $P$ of $n$ vertices, we can compute a linear-size superset $\reducedDom'$ of size $O(n)$ of guards that contains the reduced good dominators $\reducedDom$ using $O(n)$ space and $O(n \log n)$~time.
\end{theorem}

\begin{proof}
    We maintain a pair of indices $(i, j)$ as a sliding window. We either increment $i$ or increment~$j$, maintaining $\core{i-1}{j+1}$ and $\core{i-1}{j+2}$ via \slidingds with amortized $O(\log n)$ update time. 
    We also maintain a set $\reducedDom'$ subject to the following:

    \begin{invariant}
        \label{inv:reduced_good_dom}
         For any $i'< i$, or $i' = i$, and $j' < j$, if there exists a reduced good dominator for $(i', j')$ then it is in $\reducedDom'$.
    \end{invariant}

    \mysubpara{Our strategy.} 
    Let $(i,j)$ be the current index pair. 
    We apply \Cref{lem:finding_BUC} to determine whether a good dominator $(g,[u_\text{max},v_\text{max}])$ exists for $(i,j)$.
    If it exists, then \Cref{lem:finding_BUC} returns it and we add it to $\reducedDom'$. We then increment either $i$ or $j$ via the following strategy:
    \begin{itemize}[noitemsep, nolistsep]
        \item      If no good dominator exists for $(i, j)$ then we increment $i$.
        \item      If we add a guard $(g,[u_\text{max},v_\text{max}])$ to $\reducedDom'$, and $v_\text{max} = j+1$, then we increment $j$. 
        \item    Finally, if we do add a guard $(g,[u_\text{max},v_\text{max}])$ to $\reducedDom'$ and $v_\text{max} \in [j,j+1)$ then we test  whether for $(i, j+1)$ we would add a guard $(g',[u_\text{max}',v_\text{max}'])$ to $\reducedDom'$. This test happens in the exact same manner as for $(i, j)$, using the \emph{lookahead} queries for \slidingds.  If, for $(i, j+1)$, we would add a guard to $\reducedDom$, then we increment $j$. 
    Otherwise, we increment $i$.
    \end{itemize}

  \noindent
    As we only add one guard to $\reducedDom'$ each time we increment $i$ or $j$, it follows that we construct a set of good dominators of size $O(n)$ in $O(n \log n)$ total time.

    \mysubpara{Correctness.} We prove that we maintain Invariant~\ref{inv:reduced_good_dom}. 
    If we ever observe a pair $(i, j)$ for which there exists no good dominator, then by \Cref{lem:no_further_dominator} we note that there exists no good dominator for $(i, j')$ with $j' \geq j$ and we may safely increment $i$ in this case. 

    If we observe a pair $(i, j)$ for which there exists a good dominator $(g,[u_\text{max},v_\text{max}])$ and $v_\text{max} = j+1$, then for all $i' \in [i+1, j]$ the good dominator of $(i', j)$ is dominated by $(g,[u_\text{max},v_\text{max}])$ and we may safely increment $j$ in this case. 

  If there exists a dominator $(g,[u_\text{max},v_\text{max}])$ for $(i, j)$ where  
  $v_\text{max} < j+1$, then we check whether there exists a good dominator for $(i,j+1)$. If so, then any good dominator for $(i', j)$ with $i<i'< j$ will be dominated by this dominator and the invariant is thus maintained when increasing $j$. If not, then the previous argument implies that there is also no good dominator for $j' > j+1$. It follows that the invariant is maintained after incrementing $i$.
\end{proof}

% \begin{corollary}\label{cor:goodStabbing}
%     One can after $O(n\log n)$ time provide a data structure, which for given $x\in \partial P$ computes the guard $(g,[u,v])$ in $\reducedDom'$ such that $x\in[u,v]$ and $v$ is maximal, by storing them in an segment tree.
% \end{corollary}

\subsection{Computing a conforming sliding sequence \boldmath $\sigma$}\label{sec:computing_sliding_sequence}

Recall the definition of a conforming sliding sequence:

\slidingsequence*

Such a conforming sliding sequence $\sigma$ is used for all our underlying algorithms.
We show that we can compute $\sigma$ in $O(n \log n)$ time. We first show that we can check efficiently whether any guard exists for some given chain $[s,t]$ that lies inside a given visibility core. 

\begin{lemma}\label{lem:finding_guard}
    Given vertices $i$ and $j$ with $i \leq j$ and $\core{i - 1}{j + 1}$ in \slidingds, 
    let $s \in [i-1, i]$ and $t \in [j, j+1]$. 
    Given $s, t$, we can decide whether there exists a guard $(g, [s, t])$ with $g \in \core{i-1}{j+1}$ in $O(\log n)$ time. 
\end{lemma}

\begin{proof}
If $s = t$ then we simply output the guard $(s, [s, t])$. 
Otherwise, define $i^*$ as $i$ if $s < i$ and  as $i+1$ if $s = i$. 
Define $j^*$ similarly as $j$ if $t > j$ and as $j - 1$ if $t = j$.
If $i^* = j^*$ then $i^*$ can see both $s$ and $t$ and we output the guard $(i^*, [s, t])$.
Otherwise, we know that $i^* < j^*$. 
We compute a constant-size set of candidate points $\mathbb{C} \subseteq \core{i-1}{j+1}$.

\textbf{Computing the candidate set $\mathbb{C}$.}
We deploy a very similar strategy to \Cref{lem:singleDominator}.
We perform an emptiness query on $\core{i-1}{j+1}$ which, if the visibility core is non-empty, returns a point $\gamma$.
If the visibility core is empty we output that there exists no guard that guards $[s, t]$. 
We perform two ray shooting queries using the supporting ray from $i^*$ to $j^*$, and the ray from $j^*$ to $i^*$, in $O(\log n)$ time.
If both rays do not hit $\core{i-1}{j+1}$ then $\core{i-1}{j+1}$ lies either strictly right or strictly left of $\overline{i^* \, j^*}$. 
We compare $\gamma$ to this supporting line to test in constant time whether the visibility core $\core{i-1}{j+1}$ is strictly left of $\overline{i^* \, j^*}$. 
If so, then we define $g^*$ as the point in $\core{i-1}{j+1}$ that maximizes the angle $\sphericalangle(i^*, g^*, j^*)$, which we can compute in $O(\log n)$ time using an angle maximizing query. 
We can also use both ray shooting queries to detect whether the interior of $\overline{i^* \, j^*}$ intersects $\core{i-1}{j+1}$, and if it does, we define $g'$ as the last point of intersection of the open directed segment $\overline{i^* \, j^*}$ and $\core{i-1}{j+1}$.

Finally, we compute the shortest path $\short{s}{t}$ from $s$ to $t$ and test whether $\short{s}{t}$ is left-turning. 
If $\short{s}{t}$ is left-turning then we take the supporting halflines $\ell_s$ and $\ell_t$ of the first and last edge of $\short{s}{t}$ (these lines are directed away from $s$ and $t$, respectively). 
We compute $g_1$ and $g_2$: the respective last points of intersection of these halflines with $\core{i-1}{j+1}$ in $O(\log n)$ time.
The set $\mathbb{C}$ is formed by $\gamma, g^*, g', g_1$ and $g_2$ (if the respective point exists). 
For each $g \in \mathbb{C}$, we observe that $g \in \core{i-1}{j+1} \subseteq \core{s}{t}$ and so we invoke \Cref{lem:computeVisibility}
to compute whether $g$ can see $[s, t]$ in $O(\log n)$ time. 
If none of the points in $\mathbb{C}$ has this property, we output that no guard $(g, [s, t])$ exists.

\textbf{Correctness.}
Suppose that no good guard $(g,[s,t])$ with $g \in \core{i-1}{j+1}$ exists. 
Then if any guard $(g, [s, t])$ with $g \in \core{i-1}{j+1}$ exists it must be that $\sphericalangle(s, g, t) > \pi$.
This in turn implies that $\short{s}{t}$ is a left-turning chain, and so a point $p' \in \core{i-1}{j+1} \subseteq \core{s}{t}$ can see $[u, v]$ if and only if $p'$ lies right of $\ell_s$ and left of $\ell_t$. 
If there exists at least one point $p' \in \core{i-1}{j+1}$ right of $\ell_s$ and left of $\ell_t$, then one of $\{ \gamma, g_1, g_2 \}$ exists and lies right of $\ell_s$ and left of $\ell_t$.  It follows that a guard $(g, [s, t])$ exists if and only if our algorithm outputs such a guard.  
Suppose otherwise that a good guard $(g,[s,t])$ with $g \in \core{i-1}{j+1}$ exists. Then $\sphericalangle(i^*, g, j^*) \leq \pi$. We make a case distinction. 

\textbf{\boldmath Case 1: $g^*$ does not exist.} Via an identical argument as in \Cref{lem:singleDominator}, we then get that $\core{i-1}{j+1}$ must intersect the interior of $\overline{i^* \, j^*}$. In this case, any point of intersection between $\core{i-1}{j-1}$ and $\overline{i^* \, j^*}$, and in particular our precomputed point $g'$, can see $[s, t]$.

\textbf{\boldmath Case 2: $g^*$ exists.}
If $g^*$ exists, then the visibility core $\core{i-1}{j+1}$ is strictly left of $\overline{i^*\,j^*}$ by definition, and thus $\sphericalangle(i^*, g^*, j^*) < \pi$.
We now define the polygon $P_g$ that is formed by $[s, t]$ and $\overline{g \, s}$ and $\overline{t \, g}$.
The point $g^*$ by definition lies in $\core{i-1}{j+1} \subseteq \core{s}{t}$.
If $g^*$ is in $P_g$ then it lies in $P_g \cap \core{s}{t}$ and thus sees all of $[s, t]$.
We now have an identical construction to Case 2 of \Cref{lem:singleDominator} and conclude that $g^*$ can see $[s, t]$. 

\textbf{Conclusion.}
If a good guard $(g, [s, t])$ with $g \in \core{i-1}{j+1}$ exists then either $g'$ or $g^*$ can see $[s, t]$. 
If no good guard with $g \in \core{i-1}{j+1}$ exists, then any point $p \in \core{i-1}{j+1} \subseteq \core{s}{t}$ can see $[s, t]$ if and only if $p$ lies left of $\ell_s$ and right of $\ell_t$. 
So, if none of $\{ \gamma, g^*, g', g_1, g_2 \}$ guard the chain $[s, t]$, then no guard in $\core{i-1}{j+1}$ can.  
\end{proof}

\begin{theorem}\label{lem:compute_sigma1}
    There exists a linear-size conforming sliding sequence $\sigma$ and it can be computed in $O(n \log n)$ time.
\end{theorem}
\begin{proof}
    Recall that a sliding sequence has for consecutive pairs $((i, j), (i', j'))$ that $(i' - i, j' - j) \in \{ (1, 0), (0, 1) \}$. Formally, we construct $\sigma$ by focusing on a different property. 
    We define the \emph{discrete maximum} function $d_{\max}$  
    that takes any integer $i \leq n$ and returns  the maximum integer $j^*$ such that there exists a guard $(g, [i, j^*])$ with $g \in \core{i-1}{j^*+1}$. 
    Note that for any $j \in [i, d_{\max}(i)]$, the visibility core $\core{i-1}{j^*+1} \subseteq \core{i-1}{j+1}$ and so for all $j \in [i, d_{\max}(i)]$ there exists a guard $(g, [i, j])$ with $g \in \core{i-1}{j+1}$. 
    We prove that any $(i,j)$ for which there exists a guard $(g, [u, \nex{u}])$ with $u \in [i-1,i)$ and $\nex{u} \in (j,j+1]$ is included in any sliding sequence $\sigma^*$ that includes all pairs $(i, j)$ with $j \in [i, d_{\max}(i)]$ where `no strictly better' pair $(i', j')$ exists. Formally, we construct a sliding sequence $
    \sigma^*$ such that:
    \begin{itemize}
        \item $\sigma^*$ contains all $(i, j)$ with $j \in [i, d_{\max}(i)]$ where there does not exist an $i' < i$ and $j' > j$ with $j' \in [i', d_{\max}(i')]$. 
    \end{itemize}

    We claim that $\sigma^*$ is conforming. Indeed, for fixed indices $(i, j)$, if there exists a guard $(g,[u,\nex{u}])$ with $u \in [i-1,i)$ and $\nex{u} \in (j,j+1]$, then there also exists a guard $(g',[i,j])$ with $g' \in \core{i-1}{j + 1}$ and so $j \in [i, d_{\max}(i)]$.
    Thus $(i, j) \in \sigma^*$ if there exists no strictly better pair $(i', j')$.
    Suppose for the sake of contradiction that there exist indices $i' < i$ and $j' > j$ with $j' \in d_{\max}(i')$. Then by definition, there exists a guard $(g'', [i', j'])$ with $g'' \in \core{i' - 1}{j' + 1} \subseteq \core{i - 1}{j+1}$ but this contradicts the maximality of $\nex{u}$ for our original guard $g$. 

    \mysubpara{Computing $\sigma^*$.}
    We maintain $(i, j)$ subject to incrementing either $i$ or $j$ and maintain the following invariants:

    \begin{invariants}
    \begin{enumerate}[nolistsep]
        \item We store $\core{i-1}{j+1}$ in \slidingds.
        \item The sequence $\sigma^*$ that has been computed thus-far  has  $(i,j-1)$ as its last element.
        \item There exists a guard $(g,[i,j-1])$ with $g \in \core{i-1}{j}$.
        \item The sequence $\sigma^*$ is complete up to $(i, j)$.
         More formally, $\sigma^*$ computed thus-far contains all index pairs $(i_0, j_0)$ that have both of the following two properties:
         \begin{itemize}
             \item $j_0 \in [i_0, d_{\max}(i_0)]$ where there does not exist an $i' < i_0$ and $j' > j_0$ with $j' \in [i', d_{\max}(i')]$, 
             \item $i_0 < i$, or $i_0= i$ and $j_0 < j$.
         \end{itemize}
    \end{enumerate}
    \end{invariants}
    Every time we increment $i$ or $j$, we update $\core{i-1}{j+1}$ accordingly in $O(\log n)$ time. To decide whether to increment $i$ or $j$, we invoke \Cref{lem:finding_guard} using $s = i$ and $t = j$ to check if there exists a guard $(g,[i,j])$ with $g \in \core{i-1}{j+1}$. If yes, then we add the pair $(i,j)$ to $\sigma^*$ and increment $j$. Otherwise, we add the pair $(i+1,j-1)$ to $\sigma^*$ and increment $i$.
    
    \mysubpara{Correctness.} If we start our sliding window at $(i,j) = (1,2)$ then our invariants immediately imply that when we reach $i = n +1$, the sequence $\sigma^*$ have been computed correctly in $O(n \log n)$ time. What remains is to prove that the invariants always hold. Clearly, the first invariant is maintained. The second invariant is maintained by construction because before we increment $j$, we add $(i, j)$ to $\sigma^*$ and before we increment $i$, we add $(i+1, j-1)$ to~$\sigma^*$.

    Next, consider the third invariant. If we are about to increment $j$, then exists a guard $(g,[i,j])$ with $g \in \core{i-1}{j+1} \subseteq \core{i-1}{j}$, thus after incrementing $j$ the third invariant still holds. Suppose otherwise that we are about to increment $i$.
    By the third invariant, there exists a guard $(g,[i,j-1])$ with $g \in \core{i-1}{j}$.
    This guard also guards $[i+1, j]$ and $g \in \core{i - 1}{j} \subseteq \core{i}{j}$. 
    So, the guard $(g,[i,j-1])$ implies that the third invariant is maintained after incrementing $i$. 
    
    Finally, we consider the fourth invariant.
    Before we are about to increment $j$, $(i,j)$ is added to $\sigma^*$.
    Thus, if the fourth invariant holds for up to $(i, j)$ then it must hold after adding $(i, j)$ to $\sigma^*$ and incrementing $j$. 
    If we are about to increment $i$ then there does not exist a guard $(g, [i, j])$ with $g \in \core{i-1}{j+1}$. 
    Suppose for the sake of contradiction that there exists a guard $(g,[i,j'])$ with $g \in \core{i-1}{j'+1}$ and $j' > j$. However, $g \in \core{i-1}{j'+1} \subseteq \core{i-1}{j+1}$ and $[i, j] \subseteq [i, j']$ -- a contradiction. 
    We now apply the third invariant to note that there exists a guard $(g,[i,j-1])$ with $g \in \core{i-1}{j}$. The fact that there exists no guard $(g,[i,j'])$ with $g \in \core{i-1}{j'+1}$ and $j' > j$, and one guard  $(g,[i,j-1])$ with $g \in \core{i-1}{j}$, implies that 
    implies that $d_\text{max}(i) = j-1$. It follow that there is no pair $(i,j')$ with $j' \geq j$ for which $j' \in [i, d_{\max}(i')]$, and we can safely increment~$i$.
\end{proof}

\subsection{A linear-size set $\baddomsigma$ that contains all reduced good dominators in  $\baddom$}\label{sec:computing_type_1}

Consider a conforming sliding sequence $\sigma$. We define a set $\baddomsigma$ and argue that all reduced bad dominators ($\baddom$) are contained in $\baddomsigma$.
     For $(i, j) \in \sigma$ we consider $\core{i-1}{j+1}$ and note that each vertex $g$ of $\core{i-1}{j+1}$ is defined by two edges in $[i-1, j+1]$.

\begin{definition}
     We denote by $V_\sigma$ the set of all vertices of visibility cores $\core{i-1}{j+1}$ for $(i, j) \in \sigma$. 
\end{definition}

\begin{definition}
   For each $g \in V_\sigma$, defined by edges $[a - 1, a]$ and $[b, b+1]$, we define the \emph{candidate guard} $(g,[u_{\max},v_{\max}])$ where $[u_{\max},v_{\max}]$ is the maximum visible chain from $g$ that includes both $[a - 1, a]$ and $[b, b+1]$. 
   \begin{itemize}[noitemsep]
       \item We define the $\baddomsigma$ as the set of all candidate guards $(g,[u_{\max},v_{\max}])$ for $g \in V_\sigma$.
       \item  We define $\baddomsigma(i, j)$ as $\{ (g,[u_{\max},v_{\max}]) \in \baddomsigma \mid \, g \text{ is a vertex of  } \core{i-1}{j+1} \}$. 
   \end{itemize}
\end{definition}

\noindent
We compute $\baddomsigma$ by computing for all $x \in V_\sigma$ some maximal chain  $[\ell_x,r_x]$, and applying \Cref{lem:computeVisibility}.

\begin{lemma}\label{lem:computeRx}
    Let $\sigma$ be a conforming sliding window sequence, and let $x$ be a vertex of $\core{i-1}{j+1}$. 
    Define $r_x = [b, b+1]$ as the first edge with $j+1 \leq b$ where $x$ is strictly right of $r_x$. 
    We can compute $r_x$ for all $x \in V_\sigma$ in total time $O(n \log n)$.
\end{lemma}

\begin{proof}
    We iterate over all $(i, j) \in \sigma$ in order. 
    This requires $O(n)$ updates, where each update increments either $i$ or $j$. 
    Hence, we can maintain the vertices of $\core{i-1}{j+1}$ explicitly in memory using the first-in-first-out data structure from \Cref{thm:fifo}. 
    By \Cref{lem:contiguousInterval}, each point appears as a vertex of $\core{i-1}{j+1}$ exactly once, and once removed, it never reappears.
    During this update sequence, we maintain a set of \emph{active points}: points in the plane that are not necessarily current vertices of $\core{i-1}{j+1}$ but are still under consideration. 
    Whenever we advance in $\sigma$, we add all new vertices of $\core{i-1}{j+1}$ to the active set. 
    By construction, no point is ever added twice.

    Consider the moment after we increment the second index, i.e., the current pair is $(i', j')$ after increasing $j'$. 
    Denote by $H$ the half-plane left of $[j'+1, j'+2]$. 
    We remove from the active set each point $x$ that does not lie in $H$. 
    Since $x$ was a vertex of some earlier $\core{i-1}{j+1}$, where $(i, j)$ precedes $(i', j')$ in $\sigma$, it follows that $x$ remains active exactly until $r_x = [j'+1, j'+2]$. 
    It remains to show that these operations can be performed in $O(n \log n)$ total time.

    We store the active points in the data structure of Brodal and Jacob~\cite{brodal2002dynamic}, which supports insertions and deletions in $O(\log n)$ time and and extreme-point queries in $O(\log n)$ time.
    Each time we add a vertex to the active set, we perform an insertion. 
    When $j'$ is incremented, we let $H$ be the half-plane left of $[j'+1, j'+2]$ and iteratively perform extreme-point queries in the direction of the outward normal of $H$. 
    If a query returns a vertex $x$ outside $H$, then we set $r_x = [j'+1, j'+2]$ and remove $x$ from the active set, charging $O(\log n)$ time for both the query and the update. 
    The first time a query returns a vertex inside $H$, all remaining active points lie in $H$, and the process terminates.

    By \Cref{lem:core_total_size} and \Cref{lem:contiguousInterval}, at most $O(n)$ points are ever added to the active set, and each point is removed at most once. 
    As each insertion, deletion, and query takes $O(\log n)$ time, the total running time is $O(n \log n)$.
\end{proof}

\begin{corollary}\label{cor:computeLxRx}
    For every vertex $x \in V_\sigma$, we can compute the maximal chain $[\ell_x,r_x]$ of polygon edges such that the two defining edges of $x$ are in $[\ell_x,r_x]$, and $x$ is left of every supporting line of every edge in $[\ell_x,r_x]$, in total time $O(n\log n)$.
\end{corollary}
\begin{proof}
    This is an immediate consequence of \Cref{lem:computeRx} as we can iterate over $\sigma$ once forward, and once backward. 
\end{proof}

\begin{theorem}\label{thm:baddomStabbingQuery}
    For a simple polygon $P$ of $n$ vertices, and conforming sliding sequence $\sigma$, we can compute a linear-size superset $\baddomsigma$ of size $O(n)$ of guards that contains the set of reduced bad dominators $\baddom$ using $O(n)$ space and $O(n \log n)$ time.
    Furthermore, we can construct a data structure in $O(n \log n)$ time that can answer the following query in $O(\log n)$ time:
    \begin{itemize}
        \item Given $x \in \partial P$ and $(i, j) \in \sigma$ such that $x \in [i-1,i)$ return the $(g,[u,v]) \in \baddomsigma(i,j)$ that maximizes $v \in [j, j+1]$.
    \end{itemize}
\end{theorem}

\begin{proof}
    \Cref{lem:core_total_size} implies that $|V_\sigma| \in O(n)$, so $|\baddomsigma| \in O(n)$. 
    We apply \Cref{cor:computeLxRx} and compute for each $x \in V_\sigma$ the chain  $[\ell_x,r_x]$.
We then apply for each $x \in V_\sigma$ \Cref{lem:computeVisibility}, which yields $\baddomsigma$ in $O(n \log n)$ time.     
    We next prove that $\baddomsigma$ contains the set of reduced bad dominators.
    Suppose for contradiction that there is a reduced bad dominator $(c,[u_\text{max},v_\text{max}])$ that is not in $\baddomsigma$. Then $c$ is a vertex of some $\core{i-1}{j+1}$, defined by two edges $[a-1,a]$ and $[b,b+1]$, and $u_{\max} \in [i-1,i)$ and $v_{\max} \in (j,j+1]$. If $(i,j) \in \sigma$, then $(c,[u_\text{max},v_\text{max}])$ would be a candidate guard and thus in $\baddomsigma$. So, $(i,j) \notin \sigma$. By definition of $\sigma$, it must be that $\nex{u_\text{max}} >j+1$. However, by definition of the reduced bad dominators (\Cref{def:reduced_bad_dom}) there is no guard dominating $(c,[u_\text{max},v_\text{max}])$, contradicting that $\nex{u_\text{max}} >j+1$. We conclude  all reduced dominators are in $\baddomsigma$.
    
    To construct the corresponding querying data structure, we loop over all $(i, j) \in \sigma$ in order, which discretizes time $t \in [0, |\sigma|]$ where at time $t$ our loop is at $(i_t, j_t) \in \sigma$. 
    For each $g \in V_\sigma$ there is a unique time interval $[t_1, t_2]$ where for all $t \in [t_1, t_2]$, $g$ is a vertex of $\core{i_t - 1}{j_t + 1}$.
    We loop over all $t$ and use the first-in-first-out convex hull data structure from \Cref{thm:fifo} to maintain $\core{i_t-1}{j_t+1}$ explicitly in $O(n \log n)$ time and $O(n)$ space. This computes for all $g \in V_\sigma$ their time interval $[t_1, t_2]$. 
    Let $g \in V_\sigma$, $(g, [u_\text{max}, v_\text{max}]) \in \baddomsigma$ and $[t_1, t_2]$ be the corresponding interval. 
    We create a weighted rectangle $R_g := [t_1, t_2] \times [u_\text{max}, v_\text{max}]$ in $\mathbb{R}^2$ where the weight is $v_\text{max}$. 
    This creates a set $R$ of $O(n)$ weighted rectangles in the plane, which we store in a stabbing-query data structure that for any query point $q$, returns the maximum-weighted rectangle in $R$ that intersects $q$ in logarithmic time.
    Such a stabbing query data structure of size $O(n)$ can be implemented in various ways through standard techniques (for details, we refer to the stabbing query implementations in~\cite{Agarwal2005Stabbing}). 
    Given a query point $x \in \partial P$ with $(i, j) \in \sigma$, we compute in $O(\log n)$ the corresponding time $t$ such that $(i, j) = (i_t, j_t)$.
    We then perform a stabbing query with the point $(t, x)$ and the maximum-weight rectangle corresponds to the desired query output. 
\end{proof}

\section{\boldmath An $O(k n \log n)$-time algorithm}\label{sec:knlogn}

We will now present an algorithm, which for a set $X\subseteq\partial P$ will compute $\nex{x}$ for every $x\in X$ in total time $O((n+|X|)\log n)$. This then yields an $O(kn\log n)$ algorithm by applying this subroutine $k$ times, together with \Cref{thm:good_guards}. Figures \ref{fig:4sol} and \ref{fig:5sol} show how the $\nexFunc$ function is recursively applied to different starting points in the same polygon.

\begin{figure}[b]
\begin{minipage}{0.48\textwidth}
  \centering
    \includegraphics{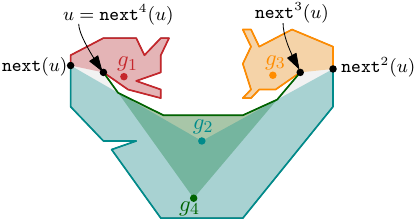}
    \caption{Illustration of optimal solution of size $4$, induced by $u\in\partial P$ that is a vertex of $P$. The guard $g_i$ is the realizing guard/dominator for $[\nexFunc^{i-1}(u),\nexFunc^ i(u)]$.}
    \label{fig:4sol}
    \end{minipage}
    \hfill
\begin{minipage}{0.48\textwidth}
    \centering
    \includegraphics{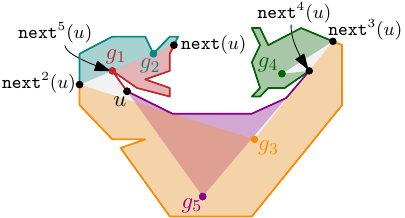}
    \caption{Illustration of non-optimal solution of size $5$, induced by $u\in\partial P$ that is a vertex of $P$. The guard $g_i$ is the realizing guard/dominator for $[\nexFunc^{i-1}(u),\nexFunc^ i(u)]$.%Good dominators are shown in green, bad dominators in red, and ugly dominators in blue.
    }
    \label{fig:5sol}
\end{minipage}
\end{figure}

\begin{lemma}\label{lem:computeNextSimultaneous}
    Let $P$ be a polygon with $n$ vertices. Let $X\subseteq \partial P$ be given with $|X| \in O(n)$. There is an algorithm which computes a guard $(g, [x, \nex{x}])$ for every $x\in X$ in total time $O((n+|X|)\log n)$ using $O(n)$ space.
\end{lemma}

\begin{proof}
    As a preprocessing step, we compute in $O(n \log n)$ time via Theorems~\ref{lem:compute_B_C}, \ref{lem:compute_sigma1} and \ref{thm:baddomStabbingQuery}:
    \begin{itemize}[noitemsep]
        \item a conforming sliding sequence $\sigma$,
        \item a linear-sized set $\reducedDom'$ of $O(n)$ guards that contains the set $\reducedDom$ of reduced good dominators,
        \item a linear-sized set $\baddomsigma$ of $O(n)$ guards that contains the set $\baddom$ of reduced bad dominators,
        \item a segment tree $T$ that stores, for all $(g, [u, v]) \in \reducedDom' \cup \baddomsigma$, the interval $[u, v]$ (with~$g$).  
    \end{itemize}

We first sort the points of $X$ along $\partial P$ in $O(|X|\log|X|)$ time and insert them into a queue. 
We then traverse all pairs $(i,j) \in \sigma$ by iteratively incrementing either $i$ or $j$. We maintain the visibility core $\core{i-1}{j+1}$ using the \slidingds. 

By definition of a conforming sliding sequence $\sigma$, any $x \in X$ with $x \in [i-1,i)$ and $\nex{x} \in (j,j+1]$ satisfies $(i,j) \in \sigma$. 
We maintain the invariant that for the current pair $(i,j) \in \sigma$, the head of the queue $x$ satisfies $i \le  x $ and $j \le  \nex{x} $. 
Since $\nexFunc$ is monotone, this invariant for the head of the queue implies the same property for all remaining elements as well. 
Whenever $\nex{x}$ has been computed, we dequeue $x$.

Let $x$ denote the current head of the queue. 
For any pair $(i',j') \in \sigma$ for which $x \in [i',i'+1)$, there exists a guard $(g,[x,j'])$ by definition of $\sigma$. 
%By the definition of $\sigma$, we $i = \lceil \nex{x} \rceil$ and $j = \lfloor \nex{x} \rfloor$. 
If $x \notin [i-1,i)$, we continue traversing $\sigma$  (incrementing either $i$ or $j$).
We are guaranteed that during this traversal, $j \leq \nex{x} $.
Once $(i, j)$ is such that $x \in [i-1,i)$, we have $\nex{x} \in (j,j+1]$ if and only if there exists no contiguous guard $(g,[x,j+1])$ with $g \in \core{i-1}{j+2}$. 
%Since $x \in [i-1,i)$, any such guard $g$ must lie in $\core{i-1}{j+1}$, which we maintain via \slidingds. 
We therefore apply \Cref{lem:finding_guard} using the lookahead queries of the \slidingds to test in $O(\log n)$ time whether a guard $(g,[x,j+1])$ with $g \in \core{i-1}{j+2}$ exists, and, if so, we increment $j$. This way, we find the pair $(i,j)$ that corresponds to $x,\nex{x}$.

Given a pair $(i,j)$ such that $x \in [i-1,i)$ and $\nex{x} \in (j,j+1]$, we compute the corresponding guard $(g,[x,\nex{x}])$ as follows. 
We first perform a stabbing query on $T$ to find a  maximum value $v$ for which there exists a guard $(g,[u,v]) \in \reducedDom' \cup \baddomsigma$ such that $x \in [u,v]$. 
This stabbing query also returns $(g, [u, v])$. 
Every good guard is dominated by a reduced good dominator, and every bad guard is dominated by either a reduced bad dominator or an ugly dominator. 
Hence, $v \ne \nex{x}$ if and only if the only guards that can see $[x,\nex{x}]$ are ugly dominators.

We compute the shortest path from $x$ to $j+1$ in $O(\log n)$ time and obtain the supporting line $\ell$ of its first edge. 
A ray-shooting query in \slidingds\ identifies the last point of intersection $g^*$ between $\ell$ and $\core{i-1}{j+1}$. 
By definition, the only ugly dominator that sees $[x,\nex{x}]$ has $g^*$ as its guard. 
We apply \Cref{lem:computeVisibility} to compute the maximal chain $[u_{\max}, v_{\max}]$ that contains $j$ and is visible from $g^*$ in $O(\log n)$ time.

It follows that $\nex{x}$ is realized either by $(g,[u,v])$ or by $(g^*,[u_{\max}, v_{\max}])$, and we can distinguish between these cases in constant time. 
We dequeue $x$ and proceed to the next element in the queue. 
Each $O(\log n)$ step is charged either to an increment of $j$ or to the removal of one element from the queue. 
Hence, the total running time is $O((n + |X|)\log n)$.
\end{proof}

\begin{theorem}\label{thm:knlogn}
    Let $P$ be an instance of the \contArt problem where $k$ denotes the size of the optimal solution.  We can compute a set of $k$ guards that guard $\partial P$ using linear space and $O(k n \log n)$ time. 
\end{theorem}
\begin{proof}
Let $X$ be the set that contains all vertices of $P$ and includes for all guards $(g, [u, v]) \in \reducedDom'$ the point $u$. \Cref{lem:compute_B_C} implies $X$ can be computed in $O(n \log n)$ time.
By \Cref{thm:good_guards} and the proof of \Cref{thm:basic_runtime}, it follows that if we apply $\nex{x}$ to each $x \in X$, recursively, $k$ times then for at least one $x \in X$ the result is a point $v > x + n$.  In other words, for this $x \in X$ the recursive application of $\nex{x}$ results in a solution of size $k$ that guards $\partial P$. 

We first compute $k$ and such a corresponding $x \in X$, and afterwards find the corresponding optimal guarding solution.
 Via \Cref{lem:computeNextSimultaneous}, we can compute $\nex{x}$ for every $x$ in $X$. We do this repeatedly, computing, and maintaining, $X_i=\{\nex{x}|x\in X_{i-1}\}$, with $X_0=X$, until there is a point $x$ such that $\nexFunc^k(x)\geq x+n$. This takes $k$ rounds, and thus a total of $O(k   n \log n)$ time, and returns the size $k$ of the optimal solution and some $x$. 
To also output the corresponding set of guards, we apply \Cref{lem:computeNextSimultaneous} once more with an input $X = \{ x \}$ where each time we find a guard $(g, [x, \nex{x}])$, we quickly add $\nex{x}$ to $X$ and the queue before we dequeue $x$.  Since \Cref{lem:computeNextSimultaneous} computes the explicit guard $(g, [x, \nex{x}])$, the result are $k$ guards that together see all of $\partial P$.  
\end{proof}

For the remainder of this paper we will assume that we know that there is no solution of size $3$ or less. This allows us to impose stronger structure, which in turn enables us to give an $O(n\log n)$ algorithm for a polygon where only solutions of size at least $4$ exist. To this end, we show that \Cref{thm:knlogn} implies an $O(n \log n)$-time test for this property:  

\begin{corollary}\label{cor:testknlogn}
    Let $P$ be a polygon consisting of $n$ vertices and $K$ be some integer. We can test in $O(K n \log n)$ time whether the \contArt problem with $P$ as its input has a solution of size $K$.   
\end{corollary}
\begin{proof}
    We simply stop the algorithm described in \Cref{thm:knlogn} after $K$ rounds.
\end{proof}

% ----- Acknowledgements (omit for double-blind submissions) -----
% \begin{acks}
% We thank ...
% \end{acks}

%\newpage
\section{Computing the functions}\label{sec:functionArrangement}
In this section, we concern ourselves with computing a representation of $\nexFunc:\partial P\rightarrow\partial P$, which we think of as a function from $[1,n + 1)$ to $[1,2n + 1]$. In particular, this representation partitions $[1,n+1)$ into $O(n)$ contiguous pieces $I_1,\ldots$, where each $I_i$ is endowed with four values $A_i$, $B_i$, $C_i$ and $D_i$ such that 
\[\nex{u} = \begin{cases}
    \frac{A_1+B_1u}{C_1+D_1u} & \text{if $u\in I_1$}\\
    \frac{A_2+B_2u}{C_2+D_2u} & \text{if $u\in I_2$}\\
     & \vdots\\
\end{cases}\]
For this, we use the sequence $\sigma$ of \Cref{def:sequence} to traverse $\partial P$, maintaining the visibility core in a \slidingds. From here on, we fix the sequence $\sigma$ as the conforming sliding sequence computed by \Cref{lem:compute_sigma1}. Notably, we only need to concern ourselves with $u$, where $\nex{u}$ is realized only by ugly dominators \uglydomref{}. Central to our analysis of the construction of this function are multiple interwoven charging arguments. In particular, we will charge almost all queries to the data structure to vertices of the shortest path from $u$ to $\nex{u}$, and vertices of the visibility core itself.

\subsection{Partitioning the function domain}

\begin{lemma}\label{lem:interstingCellInducer}
    Suppose $u\in (i-1,i)$ and $\nex{u}\in (j,j+1)$.  Let $(h,[u_h,v_h])$ be the guard in $\baddomsigma(i,j)$ such that $u\in[u_h,v_h]$, and $v_h$ is maximal. Let $e$ be the edge of $\core{i-1}{j+1}$ after $h$ in counter-clockwise order. Then $\nex{u}$ is realized by an ugly dominator \uglydomref{}, if and only if%
    \begin{enumerate}[topsep=4pt, partopsep=0pt, parsep=0pt, itemsep=0pt]
        \item[(i)] $u>u_h$, and
        \item[(ii)] $\sphericalangle(u_h,h,v_h)>\pi$, and
        \item[(iii)] the edge $e$ lies right of the first edge of the shortest path $\short{h}{j+1}$.
    \end{enumerate}%
    In this case, the realizing guard $g$ of $\nex{u}$ is on edge $e$.
\end{lemma}

\begin{proof}
    Suppose $\nex{u}$ is realized by an ugly dominator \uglydomref{} $(g,[u,\nex{u}])$. First, if $u=u_h$, then $u_h\in (i-1,i)$ and then by \Cref{cor:left-turning}, $\overline{h\,u_h}$ has a reflex-vertex in its interior. But then $h$ is the last intersection point of the supporting line of the first edge of the shortest path $\short{u}{j+1}$, and thus $g=h$, and in particular $v_h=\nex{u}$. Thus $\nex{u}$ is realized by a type \baddomref{} dominator instead. Hence $(i)$ holds.
    
    Next, by definition, $g$ is the second intersection point of the supporting line $\ell_u$ of the first edge of the shortest path $\short{u}{j+1}$ with $\core{i-1}{j+1}$. As $\nex{u}$ is realized by an ugly dominator \uglydomref{} $g$ is not a vertex of $\core{i-1}{j+1}$.  Let $C$ be the counter-clockwise chain of the boundary of $\core{i-1}{j+1}$ that is right of $\ell_u$. %Observe that for any vertex $c\in C$, the point $g$ is right of $\overline{c\,\nex{u}}$ \sarita{where do we use this?}, as otherwise a vertex of $\core{i-1}{j+1}$ could see both $u$ and $\nex{u}$, and thus $\nex{u}$ is realized by a type \baddomref{} bad guard instead. 
    Let $h'$ be the last vertex in $C$.
    Let $[u',v']$ be the maximal visibility of $h'$ in $[i-1,j+1]$. 
    This is unique, as $h'\in\core{i-1}{j+1}$ and $\short{u}{\nex{u}}$ is left turning (by \Cref{cor:left-turning}). As $h'$ is right of $\ell_u$, it can see $u$, so $u'\leq u$. %In particular, as $u\in [u',v']$ and hence $\nex{u}>v'$. 
    %As $h'$ is a vertex of $\core{i-1}{j+1}$, there is a guard $(h',[u',v'])$ in $\baddomsigma(i,j)$.
    %The chain $[i-1,j+1]$ contains the two edges defining $h'$, and then by Lemma ... $u\in [u',v']$ if $h'$ sees $u$. As $h'$ is right of $\ell_u$ it can see $u$, and thus we have that $\nex{u}>v'$.
    %\sarita{can this next sentence not just be removed?} The maximal visibility of $h'$ restricted to $[i-1,j+1]$ is contained in $[u',v']$. As $h'$ sees $u$, we have that $u\in[u',v']$, and, because $\nex{u}$ is realized by an ugly dominator \uglydomref{}, $\nex{u}>v'$.
    Furthermore, by maximality of $\nex{u}$, and the fact that $h'$ induces a bad dominator \baddomref{} whose visibility contains $[u',v']$, it must be that $g$ sees further than $v'$. It follows that the point $g$ lies right of the first edge of the shortest path $\short{h'}{j+1}$, which coincides with the supporting line of $\overline{h'\,v'}$.
    Let $x'$ be the vertex of $\short{h'}{j+1}$ defining this first edge. Any other vertex of $C$, and any other vertex of $\core{i-1}{j+1}$ that sees $u$, is left of $\overline{h'\,x'}$ and thus has $x'$ on its shortest path to $j+1$. Thus it cannot see $v'$. %It follows that $h'=h$. 
    Further, the guard $g$ must be on the edge after $h'$ in counter-clockwise order. 
    As $g$ is below the supporting line of $\overline{h'\,v'}$, so is the entirety of the edge $e$. Thus $h'$ sees its two defining edges and hence is in $\baddomsigma(i,j)$, and in particular, $h'=h$, and $[u',v']=[u_h,v_h]$, and $g$ is on $e$, and $(iii)$ holds. By general position, $x'$ is not colinear with $e$, and thus we have that both $u'=u_h$ and $v'=v_h$ lie strictly left of $e$, and thus $(ii)$ holds.

    Now, conversely, assume $(i)$, $(ii)$, and $(iii)$ hold. Then (i) implies that $u$ can see a point $g'$ on the edge $e$. By $(iii)$, this guard $g'$ sees further than $v_h$. But then it must be dominated by an ugly dominator \uglydomref{}. In particular, the next vertex (in counter-clockwise order) after $h$ of $\core{i-1}{j+1}$ can see the furthest among all guards on the edge $e$, and in particular further than $h$. Thus, $u$ does not see this next vertex, and the supporting line of the first edge of the shortest path $\short{u}{j+1}$ intersects the edge $e$, defining the the ugly dominator \uglydomref{} $g$ of $u$, concluding the proof.
\end{proof}

The consequence of the above lemma is formalized in the following lemma, stating that we can compute $O(n)$ interior-disjoint intervals such that $u\in\partial P$ is in one of these intervals if and only if $\nex{u}$ is realized by an ugly dominator \uglydomref{}. Each such interval is defined by an index $i$, such that $u\in (i-1,i)$, an index $j$, such that $\nex{u}\in(j,j+1)$, and an edge $e$ of $\core{i-1}{j+1}$, such that the realizing guard associated to it will be on edge $e$. We call the set that has these three edges associated to it $I_{i,j,e}$. Let $\mathcal{T}$ be the set of all such triples $(i,j,e)$ such that $I_{i,j,e}\neq\emptyset$. 

\begin{lemma}
    The sets $\{I_{i,j,e}|(i,j,e)\in\mathcal{T}\}$ can be computed in $O(n\log n)$ time given~$\baddomsigma$.
\end{lemma}
\begin{proof}
    Let us first consider the partition of $\partial P$ into sets $C_{i,j}=\{u\in \partial P|u\in [i-1,i),\nex{u}\in(j,j+1]\}$. Each $C_{i,j}$ is contiguous, and there are at most $O(n)$ non-empty such $C_{i,j}$. These can also be computed by \Cref{lem:core_total_size} and \Cref{lem:contiguousInterval} in $O(n\log n)$ time. Next, consider the refinement of the set $C_{i,j}$ via the arrangement of $\{[u_h,v_h]|(h,[u_h,v_h])\in\baddomsigma(i,j)\}$ into cells $C_{i,j,h}$. By \Cref{lem:contiguousInterval}, each guard in $\baddomsigma$ is in a contiguous subsequence of $\{\baddomsigma(i,j)|(i,j)\in\sigma\}$, there are again only $O(n)$ non-empty sets among all such $C_{i,j,h}$, and they can be computed in $O(n\log n)$ time via $\baddomsigma$ and $C_{i,j}$. Finally, associate every $C_{i,j,h}$ with the guard $(h^*,[u_h,v_h])\in\baddomsigma (i,j)$ (or more precisely, the edge $e$ of $\core{i-1}{j+1}$ that is after $h^*$ in counter-clockwise order) that maximizes $v_h$ subject to $C_{i,j,h}\subseteq [u_h,v_h]$. This can be done via \Cref{thm:baddomStabbingQuery}. Now consider only those $C_{i,j,h}$ such that its associated guard and edge fulfill the conditions of \Cref{lem:interstingCellInducer}. The set of cells among them associated to the same guard $(h^*,[u_h,v_h])$ is contiguous, as $i\in[u_h,v_h]$, and hence any later $C_{i,j',h'}$ is also contained in $[u_h,v_h]$. The union of these sets associated to $h^*$ and with it, associated to $e$ define the set $I_{i,j,e}$. By construction, there are at most $O(n)$ many of these which are non-empty, and these can be computed in $O(n\log n)$ time from $C_{i,j,h}$ via a simple linear scan.
    %This is an immediate consequence of \Cref{lem:interstingCellInducer}, such that for every interval of the partition of $\partial P$ into pieces corresponding to the visibilities in $\baddomsigma(i,j)$ we can check in $O(\log n)$ time if every point in its interior is dominated by an ugly dominator \uglydomref{}, and if so, on which edge $e$ its realizing guard lies, defining the triples $(i,j,e)\in\mathcal{T}$.
\end{proof}

\begin{definition}
    %Let $I_{i,j,e}$ be an interesting cell. 
    For $u\in I_{i,j,e}$ we define $x_u$ to be the vertex of $\short{i-1}{j+1}$ that defines the first edge of the shortest path $\short{u}{j+1}$. We call $x_u$ the left pivot of $u$. We similarly define the right pivot $y_u$ to be the vertex of $\short{i-1}{j+1}$ defining the first edge of the shortest path $\short{g}{j+1}$, where $g$ is the guard realizing $\nex{u}$. With this, we define the sets
    \begin{itemize}[noitemsep]
        \item $L_{i,j,e,x}=\{u\in I_{i,j,e}|x=x_u\}$, and 
        \item $R_{i,j,e,y}=\{u\in I_{i,j,e}|y=y_u\}$.
    \end{itemize}
    The collection of sets $L_{i,j,e,x}$ (and $R_{i,j,e,y}$) for all $x$ (and $y$) partition $I_{i,j,e}$. Their intersections $C_{i,j,e,x,y}=L_{i,j,e,x}\cap R_{i,j,e,y}$ together define another partition of $I_{i,j,e}$.
\end{definition}

\begin{observation}
    $L_{i,j,e,x}$ and $R_{i,j,e,y}$ are contiguous subsets of $I_{i,j,e}$ and interior-disjoint.
\end{observation}

\begin{lemma}\label{lem:computeCijexy}
    For every $I_{i,j,e}$ we have 
    \[|\{x,y|C_{i,j,e,x,y}\neq\emptyset\}|=O(|\{x|L_{i,j,e,x}\neq\emptyset\}|+|\{y|R_{i,j,e,y\neq\emptyset}\}|).\]
    Further, the set of all non-empty $C_{i,j,e,x,y}$ on $I_{i,j,e}$ can be computed in $O(|\{x|L_{i,j,e,x}\neq\emptyset\}|+|\{y|R_{i,j,e,y\neq\emptyset}\}|+ \log n)$ time.
\end{lemma}
\begin{proof}
    Let $I_{i,j,e}=[s,t]\subseteq [i-1,i]$. Let $g_s$ be the guard on edge $e$ realizing $\nex{s}$, and let $g_t$ be the guard on edge $e$ realizing $\nex{t}$. For every $u\in [s,t]$, the realizing guard $g_u$ is in $\overline{g_s\, g_t}$. The left pivot of $u$ is the first vertex of the shortest path $\short{g_u}{u}$ and in particular, by the convexity of $\short{i-1}{j+1}$ on the shortest path $\short{g_t}{s}$. Conversely, for every vertex $x$ in the shortest path $\short{g_t}{s}$ there is a $u$ such that $\ell_u=x$. Similarly, the set of right pivots for $u\in[s,t]$ is precisely the set of vertices in the shortest path $\short{g_s}{\nex{t}}$. These can be computed in time $O(\log n + |\{x|L_{i,j,e,x}\neq\emptyset\}|+|\{y|R_{i,j,e,y\neq\emptyset}\}|)$ using \Cref{ds:shortestPath}. The non-empty sets $L_{i,j,e,x}$ and $R_{i,j,e,y}$ are contiguous, and hence the non-empty sets $C_{i,j,e,x,y}$ can be computed by a simple linear scan in time $O(|\{x|L_{i,j,e,x}\neq\emptyset\}|+|\{y|R_{i,j,e,y\neq\emptyset}\}|)$ concluding the proof.
\end{proof}

\begin{lemma}\label{lem:computeFunctions}
    Let $C_{i,j,e,x,y}$ be given. Then in $O(1)$ time one can compute values $A$, $B$, $C$, and $D$ such that 
    \[\forall u\in C_{i,j,e,x,y}:\nex{u}=\frac{A+Bu}{C+Du}.\]
\end{lemma}
\begin{proof}
    For $u\in C_{i,j,e,x,y}$, $\nex{u}$ is realized by an ugly dominator \uglydomref{}. That is, its realizing guard $g$ is the intersection of the supporting line of $\overline{u\, x}$ with (the supporting line of) $e$. Similarly, given $g$, $\nex{u}$ is the intersection of the supporting line of $\overline{g\, y}$ with (the supporting line of) $[j,j+1]$. Both of these functions are M\"obius transformations, i.e., functions of the form $\frac{A+Bu}{C+Du}$, where $A$, $B$, $C$, and $D$ depend on the coordinates of the start- and end points of $[i,i+1]$, $e$, $[j,j+1]$, $x$ and $y$. As M\"obius transformations are closed under concatenation, their concatenation is of the form $\frac{A+Bu}{C+Du}$ as well, concluding the proof.
\end{proof}

Thus, if in total only $O(n)$ of the sets $L_{i,j,e,x}$ and $R_{i,j,e,y}$ are non-empty, then we can compute a representation of the $\nexFunc$ function of $O(n)$ pieces. To this end, consider the set $\hat{\mathcal{U}}$ of points $u\in\partial P$ such that $\nex{u}$ is realized by only bad guards. The realizing guard $(g,[u,\nex{u}])$ has that $\sphericalangle(u,g,\nex{u})>\pi$, and that $\short{u}{\nex{u}}$ is left-turning.  Consider the subset $\mathcal{U}\subseteq\hat{\mathcal{U}}$ for which the shortest path $\short{u}{\nex{u}}$ has at least two inner vertices. By \Cref{lem:interstingCellInducer} and \Cref{cor:left-turning}, $u\in \partial P $ such that neither $u$ nor $\nex{u}$ is a vertex of $P$ and $\nex{u}$ is realized by an ugly dominator \uglydomref{} in~$\mathcal{U}$. For a vertex $x$ of $P$, let $L_x\subseteq\mathcal{U}$ be the set of points $u\in U$ such that $x$ is the first inner vertex of the shortest path $\pi_u=\short{u}{\nex{u}}$. Let $\#L_x$ be the number of connected components of $L_x$ as a subset of $\mathcal{U}$. In the next section, we show that 
\[\sum_{\text{$x$ vertex of $P$}}\#L_x\in O(n).\]
This in turn implies that 
\[\sum_{(i,j,e)\in\mathcal{T}} \left|\left\{x|L_{i,j,e,x}\neq\emptyset\right\}\right|\leq \left(\sum_{(i,j,e)\in\mathcal{T}}2\right)+\left(\sum_{\text{$x$ vertex of $P$}}\#L_x\right)\in O(n).\]
A symmetric argument shows $\sum_{(i,j,e)\in\mathcal{T}}\left|\left\{y|R_{i,j,e,y}\neq\emptyset\right\}\right|\in O(n)$.% \sarita{sum over $y$ should be removed?}.
%\jacobus{The plan:
%\begin{enumerate}
%    \item compute coarse arrangement $C_{i,j}$
%    \item compute arrangement of vertex-induced visibilities
%    \item via \Cref{lem:type2classification} decide, which cells are actually dominated by bad guards of type 2
%    \item for these cells compute its set of left and right pivots, and with it its entire function
%    \item ignoring the first and last arrangement cell in every $C_{i,j}$, every edge of the visibility core is only considered once. $L_x$ and $R_x$ charging via section 7.3.
%\end{enumerate}
%}

\subsection{Bounding the number of pivot events}
We begin by showing a stronger version of \Cref{thm:good_guards}. For this, we first define a relation on the vertices of the polygon, where $x\prec y$, if $x$ appears before $y$ in $\pi_u=\short{u}{\nex{u}}$ for some $u\in \mathcal{U}$, which will induce a strict order.

\begin{lemma}\label{lem:totalOrder}
    Let $P$ be given. If no three contiguous guards can guard all of $\partial P$, then the transitive closure of the relation $\prec$ is a partial order.
\end{lemma}
\begin{proof}
    It suffices to show that if there is a cycle in $\prec$, i.e., there are $u_1,\ldots,u_M\in\mathcal{U}$ such that for $i < M$ the path $\pi_i=\short{u_i}{\nex{u_i}}$ has the vertex $x_i$ before $x_{i+1}$ and $\pi_M$ has the vertex $x_M$ before $x_1$, then a guarding set of size $3$ exists. Hence, let us assume such a cycle exists. Consider the prefix $x_1,\ldots,x_m,x_{m+1}$ where $x_{m+1}$ is the first $x_i$ that is more than one full revolution ahead of $x_1$, that is, $x_1\in[x_{m},x_{m+1}]$. Then $x_1,\ldots,x_m$ are sorted clockwise along the boundary of $P$, as all shortest paths $\pi_i$ are left-turning. We may assume that $x_{m+1}\in[x_1,x_2]$, by removing a prefix from $x_1,\ldots,x_{m+1}$.
    
    %Observe, that $\pi_1$ and $\pi_m$ intersect, as otherwise either $[u_1,\nex{u_1}]$ contains $[u_{m},\nex{u_{m}}$, or vice-versa, contradicting the fact that $\nexFunc$ is monotone. %otherwise $[x_1,x_2]\subseteq[x_{m-1},x_m]$, which together with the fact that each individual $\overline{\pi_i}$ is entirely left-turning, implies that $[u_1,\nex{u_1}]\supset[u_{m-1},\nex{u_{m-1}}]$

    %Then $x_1,\ldots,x_m$ are sorted clockwise along the boundary of $P$. 
    Let $\overline{\pi_i}$ the subpath of $\pi_i$ from $x_i$ to $x_{i+1}$. Observe that each $\overline{\pi_i}$ is left-turning by definition of $\mathcal{U}$. Let $P_i$ be the portion of $\partial P$ between $x_i$ and $x_{i+1}$, i.e. $P_i = [x_i,x_{i+1}]$. The concatenation of all $\overline{\pi_1},\ldots,\overline{\pi_m}$ forms a closed loop $L$ in $P$, after short-cutting $\overline{\pi_1}$ and $\overline{\pi_{m}}$ via their intersection point (they must intersect, as $x_{m}\leq x_1\leq x_{m+1}<x_2$).

    \textbf{The loop \boldmath $L$ is simple and clockwise.} First, observe that $\overline{\pi_i}$ and $\overline{\pi_{i+1}}$ intersect in $x_{i+1}$ (or the intersection point between $x_1$ and $x_m$), and as $\overline{\pi_i}$ and $\overline{\pi_{i+1}}$ are both left-turning shortest paths they cannot intersect in any other point. It remains to show that $\overline{\pi_i}$ and $\overline{\pi_j}$ do not intersect if $j\not\in\{i-1,i,i+1\}$. Assume otherwise. Then either $x_j$ or $x_{j+1}$ lies in $P_i$ contradicting the clockwise order of $x_i$, $x_{i+1}$, $x_j$ and $x_{j+1}$. Thus, $L$ is a simple loop. 
    Furthermore, as each $\pi_i$ is left-turning and contained in $P$, the loop $L$ must be clockwise.
    %As $L$ is a simple loop, it has turning number $-1$ or $1$ for any point bounded by the loop. Let $p$ be such a point. By construction, $L$ lies in $P$. Assume $L$ has turning number $1$. That is, it walks counter-clockwise around $p$. However, then $P_i$ is also in the region the plane bounded by $L$ contradicting the fact that $L$ is in $P$. Hence $L$ has a turning number $-1$, i.e., it walks around $P$ clockwise.

    \begin{figure}
    
    \centering
    \includegraphics{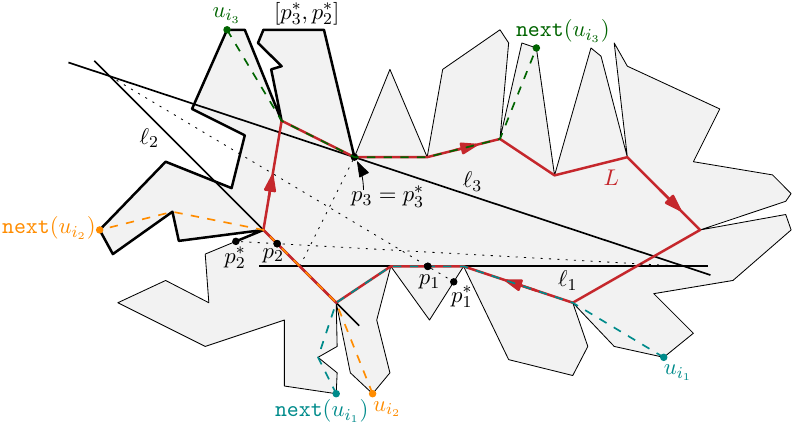}
    \caption{Construction of three guards guarding all of $\partial P$ from the proof of \Cref{lem:totalOrder}. Illustrated is the fact that $[u_{i_1},\nex{u_{i_1}}]$ contains $[p_3^*,p_2^*]$, which implies that $g_{i_1}$ sees $[p_3^*,p_2^*]$.}
    \label{fig:tripleGuard}
\end{figure}

    \textbf{At most three guards guard all of \boldmath$\partial P$.} Let $x\in L$ be a point that is in the interior of some $\overline{\pi_i}$. A half-plane $h$ defined by a line $\ell$ is said to be a tangent at $x$ if $x\in \ell$, and $h$ does not contain $\overline{\pi_i}$ in its interior. At the intersection point $x'$ between $\overline{\pi_i}$ and $\overline{\pi_{i+1}}$ we may choose (and need to specify) the membership of $x'$ to be either $\overline{\pi_{i}}$ or $\overline{\pi_{i+1}}$. %as a member of $\overline{\pi_{i+1}}$. 
    In other words, a half-plane $h$ defined by a line $\ell$ is a tangent at $x'$, if $x'\in \ell$, and $h$ does not contain one of $\overline{\pi_{i}}$ and $\overline{\pi_{i+1}}$.
    
    %\jacobus{rewrite starts here}
    Let us for now suppose, there are three points $p_1$, $p_2$, and $p_3$ on $L$, with tangents $h_1$, $h_2$ and $h_3$ defined by $\ell_1$, $\ell_2$, and $\ell_3$, such that $(i)$ each $h_i$ contains all three points $p_1$, $p_2$, and $p_3$, $(ii)$ $h_1\cap h_2\cap h_3$ is bounded, and $(iii)$ for the three $\pi_{i_1}$, $\pi_{i_2}$ and $\pi_{i_3}$ containing $p_1$, $p_2$ and $p_3$, no two are equal. Then there are three guards guarding all of $\partial P$ (refer to \Cref{fig:tripleGuard}): As $h_1\cap h_2\cap h_3$ is bounded, we can pair $p_1$ with a point $p_1^*\in \partial P$ that is in $h_{2}\cap h_{3}$, if $p_1$ is not already on $\partial P$, by shooting a ray in $P$ starting in $p_1$ away from the intersection point $\ell_{2} \cap \ell_{3}$. %This ray cannot intersect $\pi_1$ a second time and must hence be in $\P_1$. 
    As $h_1$ contains the point $\ell_2 \cap \ell_3$, and the ray cannot intersect $\pi_{i_1}$ a second time before hitting the polygon boundary, we have that $p_1^*\in P_{i_1}$, unless $\pi_1$ was already in $P_{i_1}$, in which case we define $p_1^*=p_1$. Similarly, we find for $p_2$ and $p_3$ the points $p_2^*$ and $p_3^*$. The points $p_1^*$, $p_2^*$, and $p_3^*$ appear along $\partial P$ in clockwise order. Finally, note that $p_2^*$ and $p_3^*$ are in $h_1$, and thus do not lie in $[u_{i_1},\nex{u_{i_1}}]$, which implies that $[p_{3}^*,p_{2}^*]\subseteq [u_{i_1},\nex{u_{i_1}}]$.\footnote{Recall that intervals between points on $\partial P$ are defined in counter-clockwise order.} In particular, there is a contiguous guard $(g_1,[p_3^*,p_2^*])$. Similarly, there are guards $(g_2,[p_1^*,p_3^*])$ and $(g_3,[p_2^*,p_1^*])$, which concludes the proof.

      \begin{figure}[b]
        \centering
        \includegraphics{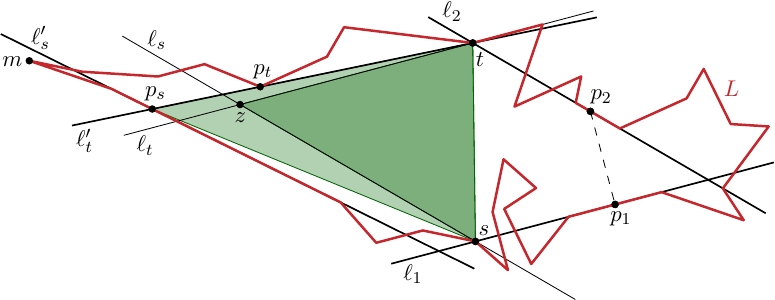}
        \caption{Construction of the points $p_s$ and $p_t$ by transforming the green triangle defined by $s$, $t$, and $z$ into the light green triangle.}
        \label{fig:triangleConstruction}
    \end{figure}

    \textbf{\boldmath The points $p_1$, $p_2$, and $p_3$ such that $(i)$, $(ii)$, and $(iii)$ hold, exist.}
    Observe that indeed any number of such points already induce a set of three such points. The only condition that might fail, by removing such points is $(ii)$, but as the intersection of tangents together form a convex, bounded set, there must also be a subset of just three such tangents whose intersection is bounded.

    We interpret $L$ as a simple polygon oriented clockwise. We begin, by picking $p_1$ and $h_1$ arbitrarily such that $p_1$ is not the intersection point of two paths $\pi_i$ and $\pi_{i+1}$. Let $p_2$ result from $p_1$ by shooting a ray from $p_1$ orthogonal to $\ell_1$. Pick $h_2$ arbitrarily, such that $p_1$ is not on $h_2$. Without loss of generality, $\ell_1$ and $\ell_2$ are not parallel, as otherwise $[\nex{u_{i_1}},u_{i_1}]$ and $[\nex{u_{i_2}},u_{i_2}]$ are disjoint, and there is even a set of guards guarding all of $\partial P$ of size $2$. Assume without loss of generality that $\ell_1$ and $\ell_2$ intersect right of the line segment $\overline{p_1\,p_2}$. As $L$ is a simple loop, there is a chain $[s,t]\subseteq L$ such that $s$ is on $\ell_1$, $t$ is on $\ell_2$, and all of $[s,t]$ is in $h_1\cap h_2$, see Figure~\ref{fig:triangleConstruction}. 
    If there is a point in $[s,t]$ with a tangent parallel to $\overline{p_1\,p_2}$, then we can choose $p_3$ as this point, as all of $[s,t]$ is left of $\overline{p_1\,p_2}$. Now assume that there is no point in $[s,t]$ with a tangent parallel to $\overline{p_1\,p_2}$.
    Consider the function $f$ which maps $x\in[s,t]$ to the distance from its orthogonal projection on $\ell_1$ to $p_1$. The function $f$ can only have one local minimum in $[s,t]$, as otherwise there is a point between two such local minima with tangent parallel to $\overline{p_1\,p_2}$. If $f$ has its local minimum at either $s$ or $t$, the chain $[s,t]$ is left-turning and there must be a point $p_3$ in $[s,t]$ whose tangent includes $p_1$ and $p_2$, as $[s,t]$ does not intersect $\overline{p_1\,p_2}$. Let instead $m\in(s,t)$ be the point where this minimum is attained.

    Now root a halfspace $h_s$ with supporting line $\ell_s$ at $s$ that is parallel to $\ell_2$ and contains $t$, and similarly root a halfspace $h_t$ with supporting $\ell_t$ at $t$ that is parallel to $\ell_1$ and contains $s$. If $[s,m]$ enters the interior of $h_s$, then the first such point $x$ has a tangent which intersects both $\ell_1$ and $\ell_2$ left of $\overline{p_1\,p_2}$, defining the desired point $p_3$. Symmetrically, we find a point $p_3$ if $[m,t]$ intersects the interior of $h_t$. So, assume that no point of $[s,t]$ enters the interior of $h_s\cap h_t$.%, as $m$ must be in neither $h_1$ nor $h_2$. 
    The point $m$ is an intersection point of two $\pi_{i_m}$ and $\pi_{i_m+1}$, as $L$ is right-turning at~$m$.

    Let $z$ be the intersection point of $\ell_s$ and $\ell_t$. Again refer to \Cref{fig:triangleConstruction} for the following construction. The triangle defined by $s$, $t$ and $z$ does not contain any point of $[s,t]$ in its interior. We transform this triangle by moving $z$ towards $m$ until it hits either $[s,m]$ or $[m,t]$. Without loss of generality, let it hit $[m,t]$. We then continue moving $z$, but away from $t$, instead of towards $m$. The result is a triangle that touches both $(s,m]$ and $[m,t)$, and still does not contain any point of $[s,t]$ in its interior. Let $p_s$ and $p_t$ be the two points where the triangle touches $(s,m]$ and $[m,t)$. If $p_s$ is $m$, then we choose the tangent $h_s'$ defined by $\ell_s'$ at $p_s$ according to $\pi_{i_m}$. If $p_t$ is $m$, then we choose the tangent $h_t'$ defined by $\ell_t'$ at $p_t$ according to $\pi_{i_m+1}$. Otherwise, we chose the tangent such that it is colinear with the face of the triangle it is touching. This way, we obtain two tangents that intersect, and their intersection contains the triangle and, in particular, $s$, $t$, $p_1$, and $p_2$. Further, the line $\ell_s'$ must intersect $\ell_1$ left of $p_1$, as otherwise there is a point in $(s,p_s)$ that is another local minimum of $f$, contradicting the fact that $m$ is the only local minimum of $f$. 
    It follows that, $p_1$, $p_2$, $p_3$, and $p_4$ fulfill properties $(i)$, $(ii)$, and $(iii)$, and in particular, a subset of at most three of them as well, concluding the proof.
    %Let now $z$ be the intersection point of $\ell_1$ and $\ell_2$. Let further $m_s\in[s,m]$ be the first point that is on the line $\ell_z$ that is parallel to $\overline{p_1\,p_2}$. Similarly, $m_t\in[m,t]$ is the last point that is on the line $\ell_z$. Any tangent at $m_s$ must contain $m_t$, and any tangent at $m_t$ must contain $m_s$, as $\overline{m_s\,m_t}$ is parallel to $\overline{p_1\,p_2}$, but no point in $[s,t]$ can have a tangent parallel to $\overline{p_1\,p_2}$. Let us now consider the chain $[s,m_s]$. It lies completely below $\ell_s$. We imagine rotating it counter-clockwise until it hits some point in $(s,m_s]$. If this point is $m_s$, then the tangent at $m_s$ contains $s$, and with it $p_1$, and all of $[m_t,t]$. Otherwise, it hits $(s,m_s)$ in some point $x$ for which the rotating line defines a at $x$ tangent which is colinear with $s$, and contains all of $[m_t,t]$. In any case, we are able to define a point 
    %Consider next the point $m_s\in [s,m]$ that is the first point on $\ell_t$. Similarly, $m_t\in[m,t]$ is the last point that is on $\ell_s$.
    %Now imagine rotating $\ell_s$ around $s$ towards $m$, until it intersects some point of $(s,m]$. Initially this is not the case, and when we reach $m$ it is the case. Suppose this point is in $(s,m)$. Then this defines a point $p_s$, whose tangent  
    %Now observe that if $[s,t]$ enters the interior of $h_s\cap h_t$, the first such point $x$ has a tangent which intersects both $\ell_1$ and $\ell_2$ left of $\overline{p_1\,p_2}$, defining $p_3$. Thus, let us assume that $[s,t]$ cannot enter $h_s\cap h_t$. 
\end{proof}

A consequence of the proof of \Cref{lem:totalOrder} is, that \emph{any} minimal solution of size at least $4$ contains at least one good guard, strengthening \Cref{thm:good_guards}. We do not make use of this fact. Instead, the consequence we are after is the following.

\begin{corollary}\label{cor:strictOrder}
    Let $P$ be a simple polygon where no three contiguous guards can guard all of $\partial P$.  Then there is strict order $\prec^*$ of the vertices of the polygon such that for any $u\in\mathcal{U}$, the first inner vertex of the shortest path $\short{u}{\nex{u}}$ is the minimum w.r.t. the order $\prec^*$, of all inner vertices in $\short{u}{\nex{u}}$.
\end{corollary}
\begin{proof}
    This is an immediate consequence of \Cref{lem:totalOrder}.
\end{proof}

%By a slight abuse of notation, we will denote by $\prec$ the transitive closure of itself.
%Let us now consider for any vertex $x$ of $P$ the sets points $I_x$, $L_x$ and $R_x$ consisting of points $u\in\mathcal{U}$ such that $x$ is an inner vertex, $x$ is the first inner vertex, and $x$ is the last inner vertex of $\pi_u$ respectively. We now show that both $L_x$ and $R_x$ are contained in some contiguous subset of $I_x$.
Let us now consider for any vertex $x$ of $P$ the subset $I_x\subseteq \mathcal{U}$ of points $u\in\mathcal{U}$ such that $x$ is on the shortest path $\short{u}{\nex{u}}$. By \Cref{cor:strictOrder}, $u\in L_x$ if and only if $x=\min_{\prec^*}\{y|y\in I_u\}$. \Cref{fig:Lx_not_connected1} illustrates that the set $L_x$ is not necessarily connected, but there is some contiguity for the set $I_x$, which we prove in the next lemma.

\begin{figure}

\begin{minipage}{0.48\textwidth}
  \centering
    \includegraphics[page=1]{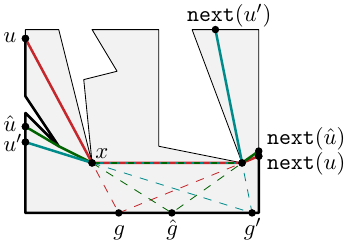}
    \caption{The set $L_x$ is not necessarily connected, but for any two points $u$ and $u'$ in $L_x$, the vertex $x$ stays on the shortest path for any $\hat{u}$ in between $u$ and $u'$, i.e., $\hat{u}\in I_x$.}
    \label{fig:Lx_not_connected1}
    \end{minipage}
    \hfill
\begin{minipage}{0.48\textwidth}
    \centering
    \includegraphics[page=2]{pictures/LxDisconnected.pdf}
    \caption{For any point $u' \in L_x$ strictly right of $\overline{x\,y}$, the potential ugly dominator~\uglydomref{} $g'$ cannot see $\nex{u}$.\\
    \textcolor{white}{hidden}}
    \label{fig:Lx_not_connected2}
\end{minipage}
\end{figure}

\begin{lemma}\label{lem:contiguoutyLx}
    Let $P$ be a simple polygon where no three contiguous guards can guard all of $\partial P$. Let $u,u'\in L_x$ for a vertex $x$ of $P$, with $u < u'$. Then at least one of $[u,u']$ and $[u',u+n]$ is contained in $I_x$.
\end{lemma}
\begin{proof}
    Recall that $u \in L_x$ means that $u \in \mathcal{U}$, so the shortest path $\pi_u=\short{u}{\nex{u}}$ is left-turning and has at least two inner vertices, and $x$ is the first inner vertex of $\pi_u$. Observe that $[u,\nex{u}]$ and $[u',\nex{u'}]$ intersect, as otherwise the two left-turning shortest paths $\pi_u=\short{u}{\nex{u}}$ and $\pi_{u'}=\short{u'}{\nex{u'}}$ cannot share a vertex.
    Assume that $u'\in[u,\nex{u}]$, and thus $\nex{u}\in[u',\nex{u'}]$, Otherwise, we can reverse the roles of $u$ and $u'$ via $u\gets u+n$.
    Refer to \Cref{fig:Lx_not_connected2} for an illustration of the proof.
    We first show that $u'$ must be left of the second edge of $\pi_u$, i.e. the edge after $x$. Suppose for contradiction that $u'$ is strictly right of the second edge of $\pi_u$. Let $y$ be the vertex after $x$ of $\pi_u$. Note that $y$ exists as $\pi_u$ has at least two inner vertices. Then $y \notin \pi_{u'}$, as $y \in \pi_{u'}$ would contradict that $\pi_{u'}$ is left-turning and contains $x$. It follows that $\nex{u'} \in [y,x]$. Because the guard $(g',[u',\nex{u'}])$ is an ugly dominator~\uglydomref{}, $g'$ is the intersection point of the supporting lines of the first and last edges of $\pi_{u'}$, and thus on the supporting line of the edge $\overline{u'\,x}$, and strictly after $x$. Hence, $g'$ is in the polygon bounded by $\overline{x\,y}$ and $[y,x]$. However, then $g'$ does not see $\nex{u}$, contradicting that it sees all of $[u',\nex{u'}]\ni \nex{u}$.
    
    We conclude that $u'$ is left of the second edge of $\pi_u$. This implies that both shortest paths $\short{u}{\nex{u'}}$ and $\short{u'}{\nex{u}}$ are also left-turning shortest paths that contain $x$. As for any point $\hat{u} \in [u,u']$ we have that $\nex{\hat{u}} \in [\nex{u},\nex{u'}]$, the shortest path $\short{\hat{u}}{\nex{\hat{u}}}$ cannot properly intersect $\short{u}{\nex{u'}}$ and $\short{u'}{\nex{u}}$. Thus, $\short{\hat{u}}{\nex{\hat{u}}}$ must be left-turning with at least two inner vertices and also have $x \in \short{\hat{u}}{\nex{\hat{u}}}$.
    %Since $u'\in[u,\nex{u}]$, and $x$ is the first inner vertex of both left-turning shortest paths $\pi_u$ and $\pi_{u'}$, we have that $u'$ lies left of the second edge of $\pi_u$,  i.e. the edge after $x$. 
    %This is the case, as otherwise any guard realizing $\nex{u'}$ lies above $\pi_u$. 
    %In particular, any such guard cannot see $\nex{u}$, contradicting the fact that $\nex{u}\in[u',\nex{u'}]$. 
    %This in turn implies that $\pi_u$ and $\pi_{u'}$ also yield the shortest paths from $u$ to $\nex{u'}$, and from $u'$ to $\nex{u}$, which are also left-turning. Thus any shortest path from some point between $u$ and $u'$ to some point between $\nex{u}$ and $\nex{u'}$ is left-turning and sandwiched between two left-turning shortest paths that go through $x$, hence the sandwiched shortest path must also go through $x$ concluding the proof.
\end{proof}

\begin{corollary}\label{cor:Ixstar}
Let $P$ be a simple polygon where no three contiguous guards can guard all of $\partial P$. Then for every vertex $x$ of~$P$ there is a contiguous subset $I_x^*\subseteq I_x \subseteq \mathcal{U}$ such that $L_x\subseteq I_x^*$. In particular, 
\[u\in L_x\iff x=\min_{\prec^*}\{y|u\in I_y^*\}.\]
\end{corollary}
\begin{proof}
    This is an immediate consequence of \Cref{cor:strictOrder} and \Cref{lem:contiguoutyLx}.
\end{proof}

\begin{lemma}\label{lem:Lxarrangement}
    Let $P$ be a simple polygon where no three contiguous guards can guard all of $\partial P$. Then
    \[\sum_{\text{$x$ vertex of $P$}}\#L_x\in O(n).\]
\end{lemma}
\begin{proof}
    Consider for every vertex $x$ of $P$ the contiguous subset $I_x^*\subseteq I_x\subseteq \mathcal{ U}$ as defined in \Cref{cor:Ixstar}. These sets define $O(n)$ interior disjoint contiguous intervals in $\mathcal{U}$, where in every interval all points in the interval are contained in the same sets $I_x^*$. In particular, $u\mapsto \min_{\prec^*}\{y|u\in I_y^*\}$ is constant for each interval. Hence for $\# L_x$, which is the number of connected components of $L_x$ as a subset of $\mathcal{U}$, we have:
    \[\sum_{\text{$x$ vertex of $P$}}\#L_x\in O(n).\qedhere\]
\end{proof}

%The statement for $R_x$ follows symmetrically. In particular, we obtain the two minimal contiguous subsets $L^*_x$ and $R^*_x$ of $I_x$ containing $L_x$ and $R_x$ respectively. Let further $\#L_x$ and $\#R_x$ be the number of contiguous subsets of $\mathcal{U}$ defining $L_x$ and $R_x$.

\begin{lemma}\label{thm:theUgly}
    Let $P$ be a simple polygon where no three contiguous guards can guard all of $\partial P$. Then
    \[\sum_{(i,j,e)\in\mathcal{T}}(|\{x|L_{i,j,e,x}\neq\emptyset\}|+|\{x|R_{i,j,e,x}\neq\emptyset\}|)\in O(n).\]
\end{lemma}
\begin{proof}
    This is an immediate consequence of \Cref{lem:Lxarrangement}, together with the dual nature of $L_x$ and $R_x$, and the fact that
    \[\sum_{(i,j,e)\in\mathcal{T}}|\{x|L_{i,j,e,x}\neq\emptyset\}|\leq \left(\sum_{(i,j,e)\in\mathcal{T}}2\right)+\left(\sum_{\text{$x$ vertex of $P$}}\#L_x\right)\in O(n).\qedhere\]
    %Traversing the polygon clockwise, we obtain all $u$
    %By \Cref{lem:totalOrder}, $\prec$ is a partial order. Thus, we can define a total order of the vertices via any order that is consistent with $\prec$. Let $<$ be this total order. With the total order at hand, observe that $u\in L_x$ iff \[x=\min_\prec\left\{\text{$y$ vertex of $P$}\middle| u\in I_y\right\}=\min_\prec\left\{\text{$y$ vertex of $P$}\middle| u\in L^*_y\right\}=\min_<\left\{\text{$y$ vertex of $P$}\middle| u\in L^*_y\right\}\] on this set $\prec$ is even a strict order. Now, as each $L^*_y$ is a contiguous subset of $\mathcal{U}$, by a simple charging argument, the minimum can change at most twice per set $L^*_y$, which in turn implies that 
    %\[\left(\sum_{\text{$x$ vertex of $P$}}\#L_x\right)\leq \left(\sum_{\text{$x$ vertex of $P$}}2\right)\in O(n).\]
    %By symmetry, we also get that $\sum_{\text{$x$ vertex of $P$}}\#R_x\in O(n)$, concluding the proof.
\end{proof}
%\subsection{Main Theorem}

\begin{theorem}\label{thm:computeFunction}
    Let $P$ be a simple polygon consisting of $n$ vertices. Suppose, $P$ cannot be guarded with three or less guards. Then a representation of $\nexFunc:[1,n]\rightarrow[1,n]$ consisting of $O(n)$ disjoint intervals, where on each disjoint interval $I$, we have that $\nex{u}=\frac{A_I+B_Iu}{C_I+D_Iu}$, can be computed in total time $O(n\log n)$.
\end{theorem}
\begin{proof}
    This is an immediate consequence of the computation of all $O(n)$ guards in $\reducedDom'$, $\baddomsigma$ in $O(n\log n)$ time together with \Cref{lem:interstingCellInducer}, \Cref{lem:computeCijexy}, 
    \Cref{lem:computeFunctions}, and \Cref{thm:theUgly}. %\sarita{we should write something here about parts that have a type (i) or good dominator realising next}
\end{proof}

\section{\boldmath An $O(n\log n)$-time algorithm}\label{sec:nlogn}

Given the presented results, we can finally present our main result. After $O(n\log n)$ time we can either output a solution of size at most $3$, or have a representation of $\nexFunc$ as a piecewise M\"obius transform, consisting of $O(n)$ pieces. Note that given this data, we can get a $(+1)$-approximation $k'$ of $k$ in $O(n)$ time by repeatedly applying $\nex{\cdot}$ to an arbitrary $x\in\partial P$ \cite[Lemma 7]{RobsonSpaldingZheng2024_AnalyticArcCover} via a linear sweep of the representation of $\nexFunc$. We finally use the techniques of \cite[Section 4]{aggarwal1989finding}, which allows us to check in $O(n\log k)$ if there is any point $x\in\partial P$ such that $\nexFunc^{k'-1}(x)\geq x+n$.
%We wish to give maximal credit to \cite[Section 4]{aggarwal1989finding} as it is able to directly resolve this subproblem. 
%However, as a reviewer has questions the validity of this application, we note that it is straightforward to find an $O(n \log n)$-time algorithm that is derived from first principles and we present this in Section~\ref{sec:firstprinciples}

\begin{restatable}{theorem}{upperboundThm}
    The \contArt problem can be solved in $O(n\log n)$ time.
\end{restatable}
% \begin{proof}
%     We first check in $O(n\log n)$ time via \Cref{cor:testknlogn}, if there is a solution of size at most three. Otherwise, via \Cref{thm:computeFunction},
%     we compute the representation of $\nexFunc$ as a piecewise M\"obius transform, consisting of $O(n)$ pieces. Via this representation, compute for an arbitrary $x\in\partial P$ the smallest value $k'$ such that $\nexFunc^{k'}\geq x+n$. By \cite[Lemma 7]{RobsonSpaldingZheng2024_AnalyticArcCover}, $k\leq k'\leq k+1$. Given this representation, and $k'$, there is an $O(n\log k)$ algorithm \cite[Section 4]{aggarwal1989finding} (we also provide a proof from first principles in Appendix~\ref{sec:firstprinciples}), which can compute a point $\hat{x}\in \partial P$, such that $\nexFunc^ {k'-1}(\hat{x})$ overtakes $\hat{x}$, i.e., $[\hat{x},\nexFunc^{k'-1}(\hat{x})]=\partial P$, or conversely, $\nexFunc^ {k'-1}(\hat{x})\geq \hat{x}+n$. If it exists, then from this $\hat{x}$, via the representation of $\nexFunc$ we can compute a solution of size $k'-1$ in time $O(k\log n)$, which is a minimum solution, as $k \leq k'\leq k+1$, implying $k'-1=k$. Otherwise, we output the solution induced by $x$, which is a minimum solution, as in this case no $\hat{x}$ exists, implying $k'=k$.
% \end{proof}
\begin{proof}
We first check in $O(n\log n)$ time via \Cref{cor:testknlogn}, if there is a solution of size at most three. Otherwise, via \Cref{thm:computeFunction},
we compute the representation of $\nexFunc$ as a piecewise M\"obius transform, consisting of $O(n)$ pieces, in $O(n \log n)$ time. Via this representation, select an arbitrary $x \in \partial P$ and we compute the smallest value $k'$ such that $\nexFunc^{k'}(x) \geq x+n$ in $O(k' \log n)$ time by repeatedly applying $\nexFunc$ in the following manner:
set $k'  = 1$.  In $O(\log n)$ time, we find the piece of $\nexFunc$ such that $x$ lies in its domain. 
We apply the function in constant time to get a value $x'$.
If $x' > x+ n$ then we return $k$ and otherwise we increment $k'$ and continue this process with $x'$.  By \cite[Lemma 7]{RobsonSpaldingZheng2024_AnalyticArcCover}, $k\leq k'\leq k+1$ for the value $k'$ that we return. 

What remains is to decide whether $k = k'$ or $k = k' - 1$. 
To this end, we  present an $O(n\log k)$ algorithm to compute a point $\hat{x}\in \partial P$, such that $\nexFunc^ {k'-1}(\hat{x})$ overtakes $\hat{x}$, or, decide no such point exists.
Formally, we find an $\hat{x}$ such that $\nexFunc^ {k'-1}(\hat{x})\geq \hat{x}+n$ or conclude that no such $\hat{x}$ exists.
In the latter case, the sequence of contiguous guards that computed $k'$ is an optimal solution.
In the former case, we compute from $\hat{x}$ an optimal solution in $O(k \log n)$ time by again recursively applying $\nexFunc$.

Our algorithm is based on the approach of~\cite[Section 4]{aggarwal1989finding}. 
Let $(I_1,f_1), \ldots, (I_m,f_m)$, with $m \in O(n)$, be the intervals and the corresponding functions partitioning $[1,2n+1)$ given by the piecewise M\"obius transform of the $\nexFunc$ function.
Given the pairs $\{ (I_i, f_i) \}$ and the set of $O(n)$ starting points $X \subset [1, n + 1)$, we wish to evaluate these functions on $X$.

\begin{properties}
These functions have two crucial properties:
\begin{enumerate}
    \item\label{prop:one} For any two functions $f_i, f_j$, let $I_i^j \subseteq I_i$ denote the interval such that for $x \in I_i^j$ it holds that $f_i(x) \in I_j$. Then the compound $f_i \circ f_j$ on $I_i^j$ is another constant-complexity function that, on a Real RAM, can be evaluated in constant time. Indeed, if the guard corresponding to $f_i$ and $f_j$ are ugly dominators then the compound of two Möbius transforms is another Möbius transform. Conversely, if $f_j$ is realized by a good or bad dominator then then $f_i \circ f_j$ is a constant $v$, where $v$ is the right endpoint of the corresponding guard $(g, [u, v])$. 
    \item\label{prop:two} The points in $X$ stay in-order when recursively applying the $\nexFunc$ function. In particular, consider some $x \in X$ and suppose that, at some point, applying $\nexFunc$ projects $x$ onto $I_j$ in between points $y \in I_j \cap X$ and $z \in I_j \cap X$. Then as we recursively apply the $\nexFunc$ to $y$, $z$, and the projection of $x$ these points will stay in the order that has the projection of $y$ first, followed by the projection of $x$, followed by the projection of $z$. 
\end{enumerate}
\end{properties}

These two properties allow us to evaluate $\nexFunc$ recursively by using lazy evaluation along a balanced binary tree that stores $X$.
In particular, partition $X$ such that $X_i$ denotes the subset of $X$ that is contained in $I_i$.

\subparagraph{Our algorithm.}
We iterate over integers $i \in [1, m]$ (note that this means we go around the polygon twice).
Throughout this iteration, we maintain an ``evaluation tree'' $T_i$.
Each leaf in $T_i$ stores a (unique) point that lies in $X_j$ for some $j < i$, and together the leaves stores all points in $\bigcup_{j < i} X_j$.
For a point $x \in X_j$, $j < i$, we then define the \emph{evaluated point} $F_i(x)$, which is the first point obtained by applying $\nexFunc$ recursively from $x$ until it exceeds the left boundary of $I_i$. We maintain the following information for each inner node $v$ of $T_i$:
\begin{enumerate}[noitemsep]
    \item A single function called a \emph{symbol.}  A symbol is either the identify function, or the compound of functions $f_j$ for $j < i$. 
    \item Real values $L(v)$ and $R(v)$ that are obtained by applying the symbols in the leaf-to-$v$ path for the leftmost and rightmost leaf in the subtree rooted at $v$, respectively.
    \item An integer $k^*(v)$ that is equal to the number of functions that have been compounded to obtain the symbol at $v$.
    \item An integer $k_{\max}(v)$, that is equal to the maximum over all leaf-to-$v$ paths of the sum of $k^*(x)$ for $x$ on that path (these paths include $v$). 
\end{enumerate}

\noindent

\begin{invariants}
We maintain the following four invariants at each iteration $i$:
\begin{enumerate}[noitemsep]
    \item\label{inv1} 
    For each leaf $x$ of $T_i$, the compound of all symbols on its leaf-to-root path, applied to $x$, yields $F_i(x)$. 
    \item\label{inv2} The leaves of $T_i$ are \emph{almost} sorted by $F_i(x)$. That is, there is a most one pair of consecutive leaves $x^*,y^*$ with $x^*$ left of $y^*$ for which $F_i(x) > F_i(y)$, for all other consecutive pairs $x,y$ with $x$ left of $y$ it holds that  $F_i(x) \leq F_i(y)$. Furthermore, if there is such a decreasing pair, then for the leftmost leaf $a$ and the rightmost leaf $b$ it holds that $F_i(a)\geq F_i(b)$. 
    \item\label{inv3} 
    The values $L(v)$, $R(v)$, $k^*(v)$, and $k_{\max}(v)$ for all nodes $v \in T_i$ are correct according to their definition.
    \item\label{inv4} For the root $r$ of $T_i$ it holds that $k_{\max}(r) \leq k' - 2$.
\end{enumerate}
\end{invariants}

\begin{figure}
    \centering
    \includegraphics{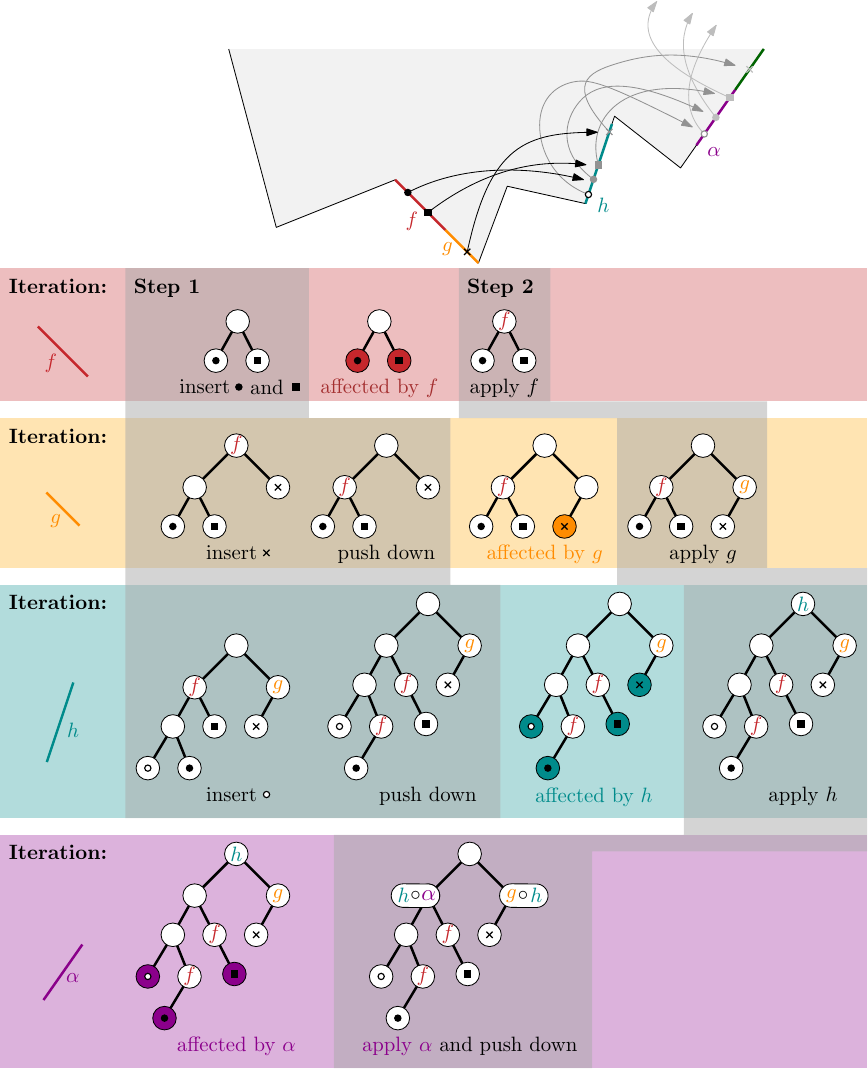}
    \caption{The figure illustrating the algorithm used to maintain our evaluation tree. Our algorithm considers the pieces of the piecewise function that is $\nexFunc$, iterating over domains corresponding to the functions $f$, $g$, $h$, and $\alpha$ in order. 
    We show the steps performed by our algorithm without applying any tree balancing. 
    }
    \label{fig:function_tree}
\end{figure}

\noindent
As our algorithm iterates over $i \in [1, m]$, we maintain these invariants in three steps:

\begin{itemize}
    \item The first step inserts all points of $X_i$ into $T$ and restores Invariants~\ref{inv1}-\ref{inv3}.
    \item The second step then applies the function $f_i$ to all leaves $x \in X$ for which $F_i(x)$ lies in $I_i$ to ensure that Invariants~\ref{inv1}-\ref{inv3} hold at the start of iteration $i+1$.
    \item 
The third step ensures that we remove any leaves $\ell$ for which the sum of $k^*(x)$ for $x$ on the path from $\ell$ to the root  equals $k' - 1$ from $T$ to maintain Invariant~\ref{inv4}. 
\end{itemize}

\noindent
Finally, we deal with tree rotations that maintain the balance of $T_i$.  

\subparagraph{Step 1: Inserting the new points.}
At the start of iteration $i$, we want to insert each point $x \in X_i$ into $T_i$ as a leaf, where $x$ is placed in-between the two consecutive leaves $a, b \in T_i$ where $F_i(a)$ precedes $x$ and $F_i(b)$ succeeds $x$. We can find the place to insert $x$ in $O(\log n)$ time using the values $L(v)$ and $R(v)$ stored in all nodes $v \in T_i$.
Starting at the root $r$ with children $u$ and $v$, we compare $x$ to $[L(u), R(u)]$ and $[L(v), R(v)]$.
Due to Property~\ref{prop:two}, exactly one of these intervals contains $x$ and we traverse into that subtree.  

Inserting any value $x \in X_i$ into $T_i$ does not break Invariant~\ref{inv2} as $F_i(x) = x$ for $x \in X_i$. 
To restore Invariant~\ref{inv1}, we must ensure that for the new leaf $x$ its root-to-leaf path contains no symbols. 
We use a ``push down'' procedure on the symbols we encounter during a root-to-leaf traversal to $x$: 
traverse the root to $x$, one node at a time, maintaining the compound $f$ of all symbols seen so far. 
As we arrive at a node $v \in T_i$, compound into $f$ the symbol at $v$. 
If $x$ lies in the left child of $v$, compound $f$ into the right child $u$ of $v$.
Update $L(u) \gets f(L(u))$ and $R(u) \gets f(R(u))$ and increment $k^*(u)$ and $k_{\max}(u)$ by adding to them the total number of functions that are compounded into $f$. 
Finally, erase the symbol at $v$ and set $k_{\max}(v) \gets k_{\max}(v) - k^*(v)$ and $k^*(v) \gets 0$. 
We then continue with the left child of $v$. If $x$ lies in the right child, we do the symmetric procedure. 

This procedure inserts a point $x \in X_i$ into $T$ in $O(\log n)$ time. Since each point is inserted exactly once, the total time of all insertions is~$O(n \log n)$.

\subparagraph{Step 2: applying the function $f_i$.}
After all points in $X_i$ are inserted into $T_i$, 
there is again a contiguous set $X'$ in $T_i$ that contains all leaves which store a point whose evaluated point $F_i(x)$ lies in $I_i$. Note that here we include the points in $X_i$ in the set $X'$, as they have already been inserted into $T_i$.
We want to (implicitly) apply $f_i$ to all these points. As these leaves form a contiguous interval, there are $O(\log n)$ subtrees in our tree that together represent $X'$. 
We want to compound the function $f_i$ to the symbols of the root nodes of these $O(\log n)$ subtrees.
We observe that the union $U$ of all paths from the root up to (but excluding) these $O(\log n)$ nodes has size $O(\log n)$.

It is crucial that the compounds represented by our symbols stay in-order. That is, in any leaf-to-root path that contains $f_i$, $f_i$ should appear as the last symbol (or, more precisely, as the last function that is compounded into the last symbol). 
To maintain this order, we again use a ``push down'' strategy across the nodes in $U$. That is, if we wish to add the symbol for $f_i$ to a node $v^*$ then we traverse the root to $v^*$, making sure that we never traverse a node in $U$ twice. 
If we encounter a node $v \neq v^*$ then we read the symbol at that node. We compound the symbol into its two children, update the $L,R,k^*,k_{\max}$ values accordingly, and erase the symbol at $v$. 
This maintains Invariants~\ref{inv1} and~\ref{inv3} for all nodes that are incident to a node in $U$, until we reach $v^*$. At $v^*$ there is now a single symbol that is the compound of the path from $v^*$ to the root. We compound $f_i$ with this symbol to satisfy Invariant~\ref{inv1} for $i+1$ and we increment $k^*(v^*)$ and $k_{\max}(v^*)$ by one. 

What remains is to restore Invariant~\ref{inv1} and~\ref{inv3} for all nodes in $U$, which no longer have a symbol.
Invariant~\ref{inv1} is correct by construction: for any leaf $\ell$ in a subtree rooted at $U$, consider the node $v^*$ on its root-to-leaf path that is incident to $u$. 
We compounded all symbols in the path from $v^*$ to the root into $v^*$ and thus invariant~\ref{inv1} is restored for $\ell$. 
We restore Invariant~\ref{inv3} by traversing $U$ in a bottom-up fashion.
Consider a node $v \in U$ with left child $a$ and right child $b$. 
As $v$ stores no symbol, it follows from Property~\ref{prop:two} that $L(v) = L(a)$ and $R(v)= R(b)$.  The value $k^*(v) = 0$ and $k_{\max}(v)$ is the maximum of $k_{\max}(a)$ and $k_{\max}(b)$.  Note that invariant~\ref{inv2} is not affected by this procedure. 
Since there are $O(\log n)$ nodes in $U$, we can charge the $O(\log n)$ time spent by this tree traversal to the interval $I_i$. Since each interval $I_i$ is charged at most once, this requires $O(n \log n)$ total time.

\subparagraph{Step 3: maintaining Invariant~\ref{inv4}.} After adding the function application of $f_i$ to $T_i$, it may be that invariant~\ref{inv4} is invalidated, i.e. for the root $r$ of $T_i$ it may be that $k_{\max}(r) > k' -2$. Observe that is must be that $k_{\max}(r) = k' -1$ in this case, as this value can only increase by one after applying $f_i$. We want to repeatedly identify a leaf $x$ in $T$ for which the sum of $k^*(x)$ for nodes $x$ on its root-to-leaf-path equals $k'-1$. We then remove this leaf from $T$. To identify such a leaf, we can follow a root-to-leaf path traversing from a node $v$ to a node $u$ if $k_{\max}(v) = k_{\max}(u) + k^*(v)$. 
If this brings us to a leaf $\ell$ storing a value $x$, we compute $F_{i+1}(x)$ in $O(\log n)$ time by applying all symbols on the leaf-to-root path to the value $x$, in-order. If $F_{i+1}(x) \geq x + n$, then we have found a solution of size $k'-1$. We output this solution and terminate. If $F_{i+1}(x) < x + n$, we remove $x$ from $T$ and any inner nodes that no longer have a leaf as descendant as a result. Finally, we recompute $k_{\max}(v)$ , $R(v)$ and $L(v)$ on this leaf-to-root path in bottom-up fashion. We repeatedly find and remove such a leaf until $k_{\max}(r)$ for the root $r$  of $T_i$ is at most $k'-2$, restoring invariant~\ref{inv4}.

\subparagraph{Tree rotations.} To rebalance the tree efficiently during the process, we need to support tree rotations in constant time. The invariants can be maintained in constant time for tree rotations in a straightforward manner. Invariant~\ref{inv2} is not affected by a rotation. To maintain invariant~\ref{inv1}, we push down the symbols at the two nodes we will perform the rotation on to their children, remove their symbols, and set their $k^*$ value to zero, before performing the rotation. To maintain invariant~\ref{inv3}, we simply recompute $L$, $R$, and $k_{\max}$ for the two affected nodes based on the invariant~\ref{inv3} values and the symbols stored in their respective children.

This procedure maintains our invariants, and thus finds whether there exists a value $\bar{x}$ for which $\nexFunc^{k' - 1}(\bar(x)) \geq \bar(x) + n$ in $O(n \log n)$ total time. Since the optimal solution has either size $k'$ or $k' - 1$, this implies the theorem.
\end{proof}

\section{Lower bound}\label{sec:lowerBound}

We complement our results in the \realRAM-model with a tight lower bound. We reduce from \setInter.  

\begin{problemstatement}[\setInter]
    Given two lists $A,B\subseteq[1,n^3]$ of integers, each of size $n$, is $A\cap B=\emptyset$?
\end{problemstatement}

\begin{theorem}\label{thm:setHardness}
In any comparison-based model of computation, the \setInter problem on inputs $A$ and $B$ with $|A| = n+1$ and $|B| = n$ takes $\Omega(n\log n)$ comparisons, even if $A$ and $B$ are integer sets and one of them is sorted.
\end{theorem}

\begin{proof}
Although this follows almost immediately from the entropy argument in~\cite{YaoLowerBound}, we present an argument for completeness.
We prove the lower bound for the special case where $A$ is sorted and all elements of $A \cup B$ are distinct.

Let $A[0] < A[1] < \cdots < A[n]$. 
We claim that any algorithm that correctly decides that $A \cap B = \emptyset$ must certify, for every element $B[i]$, in which gap of $A$ it lies. 
Indeed, suppose that after the algorithm terminates on some input, the comparison outcomes do not imply that
$
    A[t-1] < B[i] < A[t]$ for any fixed $t \in \{1,\ldots,n\}$. 
Then there exists a comparison-consistent input in which $B[i]$ is equal to one of the neighboring elements of $A$, while all comparison outcomes remain unchanged. 
The algorithm would then behave identically on an input with empty intersection and on an input with non-empty intersection, contradicting correctness.

Now consider, for every permutation $\pi$ of $\{1,\ldots,n\}$, the integer input $I_\pi$ defined by
$
    A[t] = 2t \quad\text{for } t=0,\ldots,n$ and $ B[i] = 2\pi(i)-1$ for $i=1,\ldots,n.
$
Thus
$
    A[\pi(i)-1] < B[i] < A[\pi(i)] $ for every $i\in\{1,\ldots,n\}$, and $A\cap B=\emptyset$.

For two distinct permutations $\pi_1$ and $\pi_2$, there exists an index $i$ such that $\pi_1(i)\neq \pi_2(i)$. 
Hence $B[i]$ lies in a different gap of $A$ in $I_{\pi_1}$ and in $I_{\pi_2}$. 
Since any correct certificate must determine this gap, the two inputs cannot yield the same certificate. 
Therefore there are at least $n!$ distinct certificates.
A comparison-based algorithm using $h$ comparisons can distinguish at most $2^h$ certificates. 
Hence $
    2^h \geq n!$ 
and so $h \geq \log_2(n!) = \Omega(n\log n)$.
\end{proof}

We now give a construction reducing \setInter where $A$ and $B$ are integer sets and $A$ is sorted to computing the minimum number of guards for a given polygon. For this we construct three gadgets, that we place on the boundary of a big polygon. The big polygon is a right isosceles triangle, whose base has length $2n^3$ (we refer to \Cref{fig:lowerbound}).

\begin{figure}[t]

    \centering
    \includegraphics{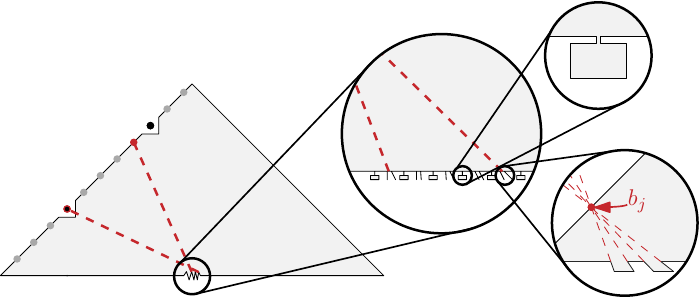}
    \caption{The lower bound construction. Black points are in the set $A$, red points in the set $B$. Shown is an instance where $A\cap B\neq\emptyset$.}
    \label{fig:lowerbound}
\end{figure}

The first gadget, called the \emph{blocker gadget} (Figure~\ref{fig:blocker}), is used to encode the sorted list $A$. We remove for each $a \in A$ a unit square centered at the boundary of the triangle at 
$(a,a)$.
Note that we can insert these gadgets along the triangle boundary in clockwise order, as $A$ is sorted. 
Due to this gadget, there there can exist no guard $(g, [u, v])$ 
where $g = (a, a)$. 
%, that is, they ensure that 
%$\not \exists g \in B_{\frac{1}{2}}(a, a)$}.
%with $g$ at the boundary of the big polygon and $g \in (a,\min(a,n^3-a))$. 

The second gadget, called the \emph{laser gadget} (Figure~\ref{fig:laser}), is used to encode the unsorted list~$B$. At the midpoint of the base, the interval from $(n^3, 0)$ to $(n^3 + n, 0)$, we will place evenly spread laser gadgets. For each $b \in B$, the laser gadget consists of two thin 
extensions to the polygon. 
Specifically, we iterate over $j \in [n]$ in increasing order, where each index $j$ has the sub interval $I_j$ of the base of the triangle from $(n^3 + 0.5 j, 0)$ to $(n^3 + 0.5j + 0.25, 0)$.  
Given $I_j$, take the lines from $(b_j, b_j)$ through four points within $I_j$,
and intersect those with some well-chosen line below the x-axis, e.g. $y=-1$, to obtain the desired vertices.
Now, to guard the base of the triangle corresponding to $I_j$, a guard must be placed exactly on the point $(b_j,b_j)$. Crucially, the placement of this gadget depends on the index $j$ and not the value $b_j$. Thus, we can specify all edges along the base of the triangle in clockwise order by iterating over $j$ from low to high. 

\begin{figure}[b]
    \centering
\begin{minipage}{0.31\textwidth}
  \centering
    %\begin{figure}
            \centering
        \includegraphics[page=1]{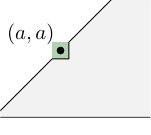}
        \caption{Blocker gadget.}
        \label{fig:blocker}
    %\end{figure}
    \end{minipage}
    \hfill
\begin{minipage}{0.31\textwidth}
    %\begin{figure}
        \centering
        \includegraphics[page=2]{pictures/gadgets.pdf}
        \caption{Laser gadget.}
        \label{fig:laser}
    %\end{figure}
\end{minipage}
    \hfill
\begin{minipage}{0.31\textwidth}
  \centering
    %\begin{figure}
        \centering
        \includegraphics[page=3]{pictures/gadgets.pdf}
        \caption{Disjoiner gadget.}
        \label{fig:disjoiner}
    %\end{figure}
    \end{minipage}
\end{figure}

The last gadget, the \emph{disjoiner gadget} (Figure~\ref{fig:disjoiner}), is used to ensure that no guard can guard more than one laser gadget. 
Between each interval $I_j$ and $I_{j+1}$ on the base of the triangle, we create a disjoiner gadget by placing a small rectangle underneath the base of the polygon, connecting it with a small channel. We also place a disjoiner gadget before $I_1$ and after $I_n$. 

\mysubpara{Output.}
The output is a simple polygon $P$, specified by its edges in clockwise order. The placement of the gadgets for $A$ depends on the values in $A$, which is allowed since $A$ is sorted. In contrast, the placement of the laser gadgets on the intervals $I_j$ and of the disjoiners between consecutive intervals depends only on the index $j$, not on the values in $B$. Thus, the clockwise order of these gadgets can be produced without sorting $B$.

\mysubpara{Robustness.}
While the above construction is not robust; i.e. to guard the polygonal extension of the interval $I_j$ a guard needs to be at exactly $(b_j,b_j)$, note that it can be made robust by widening the extensions slightly, thereby allowing the guard to be in a small ball around $(b_j,b_j)$, as long as the ball is small enough that it is entirely contained in the blocker gadget.  

\begin{lemma}\label{lem:disjoinerWork}
    No contiguous guard can guard two points $u$ and $v$ between which a disjoiner gadget lies. I.e., no contiguous guard can guard a chain $[u, v]$ where the open chain $(u,v)$ contains a disjoiner gadget. Also, no contiguous guard can guard two disjoiner gadgets at once.
\end{lemma}
\begin{proof}
    Any contiguous guard that guards from $u$ to $v$ has to guard the entirety of the disjoiner gadget. In particular, it must guard the upper right and left corner of the rectangle. Hence, the guard must be placed below the $x$-axis. Then, this guard cannot guard the points on the base of the triangle that are a small distance $\varepsilon >0$ right and left of entrance to the disjoiner gadget. Thus this guard cannot see the entire chain $[u,v]$.

    As a contiguous guard that sees an interior disjoiner gadget has to be placed inside the rectangle defining the gadget, and no guard can be placed in the rectangle of two different disjoiner gadget, no two disjoiner gadgets can be guard by the same contiguous guard.
\end{proof}

\begin{lemma}\label{lem:construction}
    The constructed polygon can be guarded by $2n+2$ contiguous guards if and only if $A\cap B = \emptyset$.
\end{lemma}
\begin{proof}
    Suppose $A\cap B=\emptyset$. Place one guard in all $n+1$ disjoiner gadgets, one guard at every $(b,b)$ for $b \in B$, and finally one guard at $(n^3,0)$. The first $n+1$ guards guard the disjoiner gadgets. The fact that $A\cap B=\emptyset$ implies that each point $(b,b)$ is contained in the polygon. Furthermore, the placement of the laser gadgets in the middle of the base implies that the visibility from $(b,b)$ to the corresponding laser gadget (including the edge of the base to the preceding disjoiner gadget and the edge to the following disjoiner gadget) is visible to $(b, b)$. So, the second $n$ guards guard each of the laser gadgets and the adjacent edges. Finally, the guard at $(n^3,0)$ guards the entire rest of the polygon boundary, that is, from the last disjoiner gadget up to the first disjoiner gadget in counter-clockwise direction.

    Now suppose, a solution of size $2n+2$ exists. By \Cref{lem:disjoinerWork}, at least $n+1$ of those must guard the $n+1$ areas of the polygon that are separated by disjoiner gadgets. That leaves $n+1$ unguarded disjoiner gadgets, which each require their own guard. Hence each area that is separated by a disjoiner gadget has to be guarded by a single guard. In particular, every laser gadget has to be guarded by a single guard. But this is only possible, if $(b,b)$ is in the polygon for every $b\in B$. This in turn is only possible, if $B\cap A=\emptyset$, concluding the proof.    
\end{proof}

\begin{theorem}
    Given a polygon $P$, and a integer $k$, deciding whether $\partial P$ can be guarded by $k$ guards, requires $\Omega(n\log n)$ time in the worst case.
\end{theorem}
\begin{proof}
    This is an immediate consequence of \Cref{thm:setHardness} and \Cref{lem:construction}, together with the fact that the construction from \Cref{lem:construction} takes $O(n)$ time, and the size of the coordinates of the vertices is polynomial in the input size.
\end{proof}

\begin{remark}
    The hardness does not depend on the fact that the construction is in general position. In fact, the instance remains hard, even after perturbing every coordinate slightly.
\end{remark}

\section{Conclusion}

We presented matching upper and lower bounds for the $\contArt$ problem.
A natural direction for further research is to investigate whether our algorithmic ideas extend to polygons with holes. For polygons with holes, there are two natural variants: either one must guard the boundaries of the holes in addition to the outer boundary, or only the outer boundary is required to be guarded. In both cases, at least the outer boundary must be covered. For our approach, however, several crucial components break down when holes are allowed. Not only do key data structures, such as the shortest-path structure, no longer behave as required, but certain structural properties also fail to generalise. In particular, \Cref{thm:good_guards} shows that for simple polygons any solution can be modified so that it contains a good guard. Given the $O(n)$ reduced good dominators, implementing $\nexFunc$ yields the $O(kn \log n)$ algorithm. In polygons with holes, by contrast, there exist configurations in which no such modification is possible (see \Cref{fig:hardWithHoles}). The essential insight behind the proof of \Cref{thm:good_guards} is that every bad guard is bounded by reflex vertices of the outer boundary, which themselves must be guarded. If instead the bounding vertices occur on holes, we can no longer force the presence of a good guard. It is therefore likely that new structural insights are required to obtain near-linear-time algorithms for polygons with holes, if near-linear time is achievable at all.

Another direction concerns the model of computation. Throughout this work we have ultilised the $\realRAM$ model, which is standard in computational geometry. One may ask how our results behave under different models of computation. For the classical $\artGall$ problem, Abrahamsen, Adamaszek, and Miltzow~\cite{abrahamsen_et_al:LIPIcs.SoCG.2017.3} showed that even when polygon vertices are integers, an optimal solution may require guard positions with irrational coordinates and the \realRAM is therefore arguably a necessary model of computation. The $\contArt$ problem does not exhibit this phenomenon. In fact, the existing solutions~\cite{theEnemy, RobsonSpaldingZheng2024_AnalyticArcCover} imply a mild upper bound on bit complexity: if $P$ is described using $\log n$ bits per coordinate, then any optimal solution requires at most $\tilde{O}(kn)$ bits. A formal analysis is given in \cite{RobsonSpaldingZheng2024_AnalyticArcCover}.

Switching computational models affects both the lower and upper bounds. In any comparison-based model, the $\Omega(n \log n)$ lower bound remains valid, but in models such as the $\wordRAM$, this lower bound no longer applies. On the upper-bound side, in models where the cost of function evaluation depends on input bit length, our algorithm incurs additional overhead. Our $O(n \log n)$ algorithm computes the $k$-fold composition of Möbius transforms. After $O(k)$ preprocessing, the resulting composition can be evaluated in $O(1)$ time in the $\realRAM$ model. In the $\wordRAM$ model, however, the bit complexity of the composed transform is linear in $k$, since the parameters become products of the coordinates of $k$ input points. Thus, if the input is representable with $n$ bits, the running time increases to roughly $O(kn \polylog n)$. This raises the question of whether matching upper and lower bounds for $\contArt$ can also be obtained in the $\wordRAM$.
We note that in simple polygons, the bit-complexity required to represent a solution is frequently a problem. 
For example, Kahan and Snoeyink~\cite{Kahan1996BitComplexity} show that there exists a polygon $P$ with points $u, v \in P$ such that the minimum-link path from $u$ to $v$ requires $\Theta(n^2 \log n)$ bits to represent even then $P$, $u$ and $v$ have small rational coordinates.
The construction is to make $P$ a spiralling path, where the complexity of the minimum-link path `compounds' to require $\Theta(k)$ bits after $k$ edges,

While we do not provide a formal argument, it appears unlikely that any algorithm in the $\wordRAM$ model can achieve running time $O(k^{2-\varepsilon})$. The $k$-fold composition of $\nex{\cdot}$ may contain intervals on which $\nexFunc^k(u)$ is expressed by a single Möbius transform whose bit complexity is $\Theta(k)$ (see \Cref{fig:kSquared}). This occurs whenever a long subsequence of the form $[u,\nex{u}], [\nexFunc(u),\nexFunc^2(u)], \ldots, [\nexFunc^{k-1}(u),\nexFunc^{k}(u)]$ of length $\Theta(k)$ is realized by ugly dominators. Consequently, a solution may require $\Theta(k^2)$ bits. We believe that one can construct a polygon $P$ that resembles the spiraling path as in~\cite[Figure~1]{Kahan1996BitComplexity} where the optimal solution can only be represented by using $\Omega(k^2)$ bits.
The key idea is to have a spiraling path with a gadget at the start and the end, that forces the optimal solution to have a sequence of ugly dominators that coincide with the minimum-link path from $u$ to $v$ (see \Cref{fig:kSquared}). 
In particular, we believe that one can combine ideas from \Cref{fig:4sol,fig:5sol} with \Cref{fig:kSquared} to construct instances where every optimal solution has bit complexity $\Omega(k^2)$. We do not provide the formal argument and include only this sketch to support the claim that we find it likely that in the $\wordRAM$ model, the $\contArt$ problem inherently requires $\Omega(kn)$ time.

\begin{figure}[H]

\begin{minipage}{0.48\textwidth}
    \centering
    \includegraphics{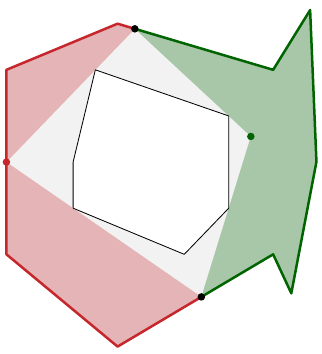}
    \caption{Illustration of a set of guards guarding the outer boundary of a polygon $P$ with a hole. The solution cannot readily be transformed via \Cref{thm:good_guards} to have at least one good guard.}
    \label{fig:hardWithHoles}
    \end{minipage}
    \hfill
\begin{minipage}{0.48\textwidth}
    \centering
    \includegraphics{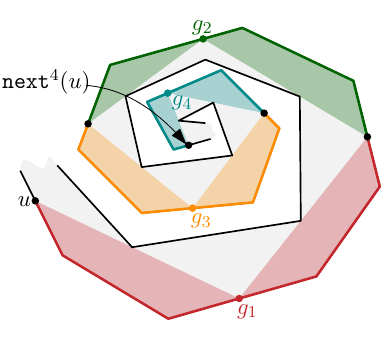}
    \caption{Illustration of compounding bit complexity in the \wordRAM-model. The bit complexity of $\nexFunc^ k(u)$ may be linear in $k$. The bit complexity of a minimal set of guards may be quadratic in $k$.}
    \label{fig:kSquared}
\end{minipage}
\end{figure}

\bibliography{refs}

\end{document}